\documentclass[pdftex,pra,aps,onecolumn,twoside,nofootinbib,showkeys]{revtex4}
\pdfoutput=1
\usepackage[utf8]{inputenc}
\usepackage[dvips]{graphicx}
\usepackage[stable]{footmisc}
\usepackage{amsmath,amsthm,amssymb}
\usepackage{youngtab}
\usepackage[all]{xy}
\usepackage{dsfont}
\newcommand{\<}{\langle}
\renewcommand{\>}{\rangle}
\newcommand\be{\begin{equation}}
\newcommand\ee{\end{equation}}
\newcommand\ot{\otimes}

\newcommand\cH{\mathcal{H}}

\newcommand\sH{\operatorname{Hom}}

\newcommand{\opV}{\operatorname{V}}
\newcommand{\mult}{\operatorname{mult}}
\newcommand{\opX}{\operatorname{X}}

\newcommand\hcal{\cH}
\newcommand\bea{\begin{array}}
\newcommand\eea{\end{array}}
\newcommand\ben{\begin{eqnarray}}
\newcommand\een{\end{eqnarray}}

\newcommand\bei{\begin{itemize}}
\newcommand\eei{\end{itemize}}
\newcommand\bee{\begin{enumerate}}
\newcommand\eee{\end{enumerate}}

\usepackage{marginnote}

\newcommand{\tr}{{\rm Tr}}
\def\ot{\otimes}

\def\bei{\begin{itemize}}
\def\eei{\end{itemize}}

\newtheorem{theorem}{Theorem}
\newtheorem{fact}[theorem]{Fact}
\newtheorem{lemma}[theorem]{Lemma}
\newtheorem{corollary}[theorem]{Corollary}
\newtheorem{example}[theorem]{Example}
\newtheorem{proposition}[theorem]{Proposition}
\newtheorem{definition}[theorem]{Definition}
\newtheorem{notation}[theorem]{Notation}
\newtheorem{remark}[theorem]{Remark}

\newtheorem*{rep@theorem}{\rep@title}
\newcommand{\newreptheorem}[2]{%
\newenvironment{rep#1}[1]{%
 \def\rep@title{#2 \ref{##1} (restatement)}%
 \begin{rep@theorem}}%
 {\end{rep@theorem}}}
\makeatother

\newreptheorem{theorem}{Theorem}
\newreptheorem{lemma}{Lemma}
\newreptheorem{definition}{Definition}

\newcommand{\opP}{\operatorname{P}}
\newcommand{\opE}{\operatorname{E}}
\newcommand{\opY}{\operatorname{Y}}
\newcommand{\opZ}{\operatorname{Z}}
\newcommand{\w}{\widetilde{v}}
\newcommand{\ww}{\widetilde{\omega}}

{\par\addvspace{\medskipamount}\noindent\textbf{Examples.}\hspace{1ex}}%
{\par\medskip}

\DeclareUnicodeCharacter{981}{\phi}
\DeclareUnicodeCharacter{348}{{\hat{S}}}
\DeclareUnicodeCharacter{8224}{{\dagger}}

\usepackage{color}

\begin{document}
\reversemarginpar

\title{Commutant structuture of $ \text{U}^{\ot (n-1)}\ot \text{U}^*$ transformations}

\author{Micha{\l} Studzi{\'n}ski$^{1,2}$, Micha{\l} Horodecki$^{1,2}$ and Marek Mozrzymas$^3$}

\affiliation{
$^1$Institute for Theoretical Physics and Astrophysics,
University of Gda{\'n}sk, 80-952 Gda{\'n}sk, Poland\\
$^2$National Quantum Information Centre of Gda\'{n}sk, 81-824 Sopot, Poland\\
$^3$Institute for Theoretical Physics,
University of Wroc{\l}aw, 50-204 Wroc{\l}aw, Poland}

\date{\today}


\begin{abstract}
In this paper we have found irreducible representations (irreps) of the algebra of partially transposed permutation operators on last subsystem. We give here direct method of irreps construction for an arbitrary number of subsystems $n$ and local dimension $d$. Our method is inspired by representation theory of symmetric group $S(n)$, theory of Brauer Algebras and Walled Brauer Algebras. 
\end{abstract}


\maketitle 

\section{Introduction}

Partial transposition plays important role in quantum information theory. It gives us necessary and sufficient conditions for separability for bipartite systems in $2 \times 2$ and $2 \times 3$ dimension case~\cite{Horodecki1, Peres1} and it is called PPT criterion (\textsl{Positive Partial Transpose}). Namely if spectrum of partially transposed bipartite density operator is positive our state is separable. For other cases (higher dimensions and multiparty case) partial transposition gives only necessary condition for separability of states, but still is very strong tool. Separability problem has been intensively studied also of multipartite systems~\cite{Horodecki2}. There are many papers about analysis of PPT properties for states with special symmetries. In particular Eggeling and Werner in~\cite{Eggeling1} present result on separability properties for tripartite states which are $U^{\ot 3}$ invariant using PPT  property and tools from group theory for an arbitrary dimension of subsystem space. 
In~\cite{Tura1, Augusiak1} authors present solution on open problem of existence of four-qubit entangled symmetric states with positive partial transposition and generalize them to systems consisting an arbitrary number of qubits. In particular, they provide criteria for separability of such states formulated in terms of their ranks. PPT property turned out also relevant for a problems in computer science: it is relaxation of some complexity problem, which can be written in terms of separability~\cite{Montanaro1,Brandao1}.

As one knows, if the system possess some symmetries, it can simplify a lot analysis of its properties. A natural tool is here group representation theory:  e.g. operators which are $U^{\ot n}$ invariant can be decomposed  by means of Schur-Weyl duality (see~\cite{Audenaert2006-notes},~\cite{RWallach} and Section~\ref{sec:young} of our work) into irreducible representations (irreps) of symmetric group $S(n)$. Then finding spectra by means of such decomposition becomes much more efficient and easier.
Let us mention here that Schur-Weyl duality and concept of Schur basis, were successfully applied in quantum information, in particular, it was used for qubit purification problems~\cite{springerlink:10.1007/PL00001027,PhysRevLett.82.4344,PhysRevA.84.022320}, efficient quantum circuits~\cite{Harrow1,Bacon1}, distillation of quantum entanglement~\cite{Czechlewski1}, cloning machines~\cite{Cwiklinski1}, local quantum circuits~\cite{Brandao2} and one-to-one relation between the spectra of a bipartite quantum state and its reduced states~\cite{Christandl1}.
   
In almost all papers mentioned in the last paragraph authors have used commutant structure of $U^{\ot n}$ transformation.
Next step can be analysis a commutant of $U^{\ot (n-k)}\ot (U^*)^{\ot k}$ invariant operators, where by $*$ we denote complex conjugation. This corresponds exactly with the case when partial transposition is applied on last $k$ subsystems.  We know that when some operator is invariant under such class of transformation it can be decomposed in terms of partially transposed permutation operators, so it is enough to investigate properties algebra  of partially transposed permutation operators on last $k$ subsystems. In this paper we present an analysis when $k=1$ and for this case we will denote our algebra by $\mathcal{A}$. One can see that it is  generalization to multipartite case  of commutant structure for  $U\ot U^*$ transformations, where the commutant is spanned by bipartite identity operator \noindent
\(\mathds{1}\) and maximally entangled state between two systems $\Phi^+$~\footnote{Maximally entangled state in computational basis is given by $|\Phi^+\>=\frac{1}{\sqrt{d}}\sum_i|ii\>$, where $d$ is dimension of Hilbert space.},~\cite{Werner,Eggeling1}).

There is well known connection between algebra $\mathcal{A}$ of partially transposed permutation operators $\opV^{\Gamma}(\pi)$~\footnote{by $\Gamma$ we denote partial transposition over an arbitrary number of subsystems. In our paper we have $\Gamma=t_n='$}, which are $U^{\ot (n-1)}\ot U^*$ invariant with Brauer Algebras~\cite{Gavarini1, Brauer1,Pan1} and its subalgebras so called Walled Brauer Algebras~\cite{Koike1,Turaev1}.  Namely algebra of all partially transposed permutation operators together with those operators on which partial transposition acts non trivially is a representation of Brauer Algebra. When we consider only algebra $\mathcal{A}$ we know that it is a representation of Walled Brauer Algebra (see for example~\cite{Zhang1}). However it appears that the representation is isomorphic to (Walled) Brauer Algebra, only  for particular relation between number of subsystems $n$ and their dimension $d$. It is easy to see this comparing their dimensions. From~\cite{Brundan1} we know that dimension of Walled Brauer Algebra is equal to $n!$. On the other hand we know that dimension of algebra $\mathcal{A}$ is also equal to $n!$, whenever $d> n-1$. In this case, our two algebras are isomorphic. In the case $d \leq n-1$ some elements of operators basis of $\mathcal{A}$ are linearly dependent~\footnote{For example projector onto antisymmetric subspace is equal zero whenever $d=n-1$.}, so $\operatorname{dim}\mathcal{A}<n!$ while dimension of Walled Brauer Algebra is still equal to $n!$, so we  do not have above-mentioned isomorphism. 

One important implication of lack of isomorphism is the issue of semisimplicity.
Translating necessary and sufficient condition from~\cite{Cox1} into language of number of systems and local dimensions of the Hilbert space we obtain that Walled Brauer Algebra is semisimple if and only if $d>n-2$ and also from the same work we know how to label irreducible components. 
For $d=n-1$ both the algebra $\mathcal{A}$ and Walled Brauer Algebra though not isomorphic anymore, are still both semisimple. When condition $d>n-2$ is not fulfilled, then our algebra $\mathcal{A}$ is still semisimple while Walled Brauer Algebra is not.

In this paper we are interested in the problem of decomposition of algebra $\mathcal{A}$ into irreducible components. The new result is the construction of biorthogonal basis in every irreducible space, calculate its dimension and finally we present formulas for  matrix elements in every irrep. We present full solution of the problem for the regime when $d>n-2$. In the other case, i.e. when $d \leq n-2$ we show that algebra $\mathcal{A}$ is still simple reducible in contrary to Walled Brauer Algebra. In this regime we present sketch of the full solution and formulate problems connected with this construction.

This work is organized as follows. In Section~\ref{sec:young} we give brief introduction to Schur-Weyl duality which is one of the main motivation for this work. Reader can see how to decompose unitary invariant operators (for example permutation operators ) into irreducible representation of symmetric group. In Section~\ref{OurResults} we presents all crucial definitions and main results without explicit proofs and explanations. We define nonorthogonal  basis for every invariant space and we show that every such space is indeed irreducible. Section~\ref{construct} is divided into three parts. In the first part we present matrix representations for the case when $d>n-2$ and for small values of $n$. In the second part we show difference between our algebra $\mathcal{A}$ and Walled Brauer Algebra for $d\leq n-2$- we give example that algebra $\mathcal{A}$ is still semisimple. We also formulate there a new mathematical problem strictly connected with our construction. At the end of this section we give method how to calculate multiplicity of given irrep. This work contains also Appendix which includes Subsection~\ref{FullDescription} where we give proof of all lemmas nad theorems from Section~\ref{OurResults}. In Appendix we present also discussion about operator bases, matrix elements of irreps and we give an exemplary application of our construction.

\section{Preliminaries}
\label{sec:young}
Consider a unitary representation of a permutation group $S(n)$ acting on the $n-$fold tensor product of complex spaces $\mathbb{C}^d$, so our full Hilbert space is $\mathcal{H}\cong (\mathbb{C}^d)^{\otimes n}$. For a fixed permutation $\sigma\in S(n)$ a unitary transformation $\opV(\sigma)$ is given by
\be
\label{unitary}
\opV(\sigma)\left( |i_1\>\otimes \ldots  \otimes |i_n\>\right)=|i_{\sigma^{-1}(1)}\> \otimes  \ldots  \otimes |i_{\sigma^{-1}(n)}\>,
\ee
where $|i_1\>,\ldots,|i_n\>$ is a standard basis in $(\mathbb{C}^d)^{\otimes n}$ . The space of rank$-n$ tensors can be also consider as a representation space for a general linear group $\operatorname{GL}(d,\mathbb{C})$. Let $U\in \operatorname{GL}(d,\mathbb{C})$, thus, this  induces in the tensor product $(\mathbb{C}^d)^{\otimes n}$ the following transformation
\be
\label{unit}
U^{\ot n}\left(|i_1\>\otimes  \ldots  \otimes |i_n\>\right)=U|i_1\> \otimes  \ldots \otimes U|i_n\>.
\ee
A key property is that these two representations turn out to be each other commutants.  Any operator  on $(\mathbb{C}^d)^{\otimes n}$ that commutes with all $U^{\otimes n}, \ \forall U\in \operatorname{GL}(d,\mathbb{C})$, is a linear combination of permutation matrices $\opV(\sigma)$. Conversely, any operator commuting with all permutation matrices $\opV(\sigma), \ \forall \sigma\in S(n)$, is a linear combination of $U^{\otimes n}$. This duality is called Schur-Weyl duality (see, for example,~\cite{Audenaert2006-notes},~\cite{RWallach}). It was shown  that there always exists some basis called the Schur basis which gives decomposition of $\opV(\sigma)$ and  $U^{\otimes n}$ into irreducible representations~\cite{Fulton},~\cite{RWallach} (irreps) simultaneously. Thanks to this, the space $(\mathbb{C}^d)^{\otimes n}$ can be decomposed into irreps of $S(n)$
\be
\label{decomp}
(\mathbb{C}^d)^{\otimes n}\cong \bigoplus_{\alpha \vdash n} \mathcal{H}_{\alpha}^{\mathcal{U}}\otimes \mathcal{H}_{\alpha}^{\mathcal{S}},
\ee
where $\alpha$ labels inequivalent irreps of $S(n)$ and $\mathcal{H}_{\alpha}^{\mathcal{U}}$ is the multiplicity space. It is called Schur-Weyl decomposition. The labels $\alpha$ can be interpreted also like partitions of some natural number $n$ and are denoted by $\alpha \vdash n$ . Every partition is a sequence $\alpha=(\alpha_1,\ldots,\alpha_r)$ satisfying
\be
\label{partition}
\alpha_1\geq \alpha_2\geq \ldots \geq \alpha_r \geq 0,\quad \sum_{i=i}^r\alpha_i=n,
\ee
where $r\in \{1,\ldots,n\}$. Such partitions correspond to some diagram, which is called the Young diagram \cite{Fulton}. Here are few examples of Young diagrams for $n=4$:
\[
\begin{array}{ccccc}
\yng(4) & \yng(3,1) & \yng(2,2) & \yng(2,1,1) & \yng(1,1,1,1)\\
\alpha=(4), & \alpha=(3,1), & \alpha=(2,2), & \alpha=(2,1,1), & \alpha=(1,1,1,1)\\
\end{array}
\]
For example the permutation operator $\opV(\sigma)$ can be decomposed, due to Schur - Weyl decomposition, in the following way:
\begin{equation}
\label{eq:swapdec}
\opV(\sigma)=\bigoplus_{\alpha} \text{\noindent
\(\mathds{1}\)}_{r(\alpha)} \ot \widetilde{\opV}_{\alpha}(\sigma),
\end{equation}
where $\sigma \in S(n)$ and $r(\alpha)$ is the dimension of a unitary part.
The operators $\widetilde{\opV}_{\alpha}(\sigma)$ are irreducible representations
of some $\sigma \in S(n)$.
 Thanks to the above-mentioned method, we can decompose $U^{\otimes n}$-invariant state (see for example~\cite{Eggeling1}) in the following way:
\begin{equation}
\label{general2}
\rho_{1\ldots n}=\bigoplus_{\alpha} \text{\noindent
\(\mathds{1}\)}_{r(\alpha)}\ot \widetilde{\rho}_{\alpha}.
\end{equation}
In next sections we explain how to construct similar decomposition of Hilbert space for operators which are $U^{\ot (n-1)}\ot U^*$ invariant. We show also how to calculate matrix elements of irreps for partially transposed permutation operators $\opV(\sigma)$.

\section{An Analysis of Commutant Structuture for $ \text{U}^{\ot (n-1)}\ot \text{U}^*$ Transformations}
\label{OurResults}
In this paper we give an analysis of algebra $\mathcal{A}$ based on representation theory for permutation group $S(n)$ (see previous section and cited papers). We show how to
construct irreducible representations (irreps) of algebra of partially transposed permutation operators over last $n^{\text{th}}$ subsystem, i.e. we are analyzing commutant structure of $U^{\ot (n-1)} \ot U^*$ operations.

As we mentioned in the introduction we can split our considerations into two main parts. First part is connected with the case when $d> n-1$. In this regime algebra $\mathcal{A}$ of partially transposed swap operators is isomorphic to semisimple Walled Brauer Algebra. From previous papers (see~\cite{Cox1}) we know how to decompose an abstract Walled Brauer Algebra into irreducible components. Our main result in this case is explicit method of irreps construction (we find their matrix elements) for  partially transposed swap operators on last subsystem. Second part treats on the case when $d \leq n-1$. From previous section we know that for $d=n-1$ both algebras are semisimple but they are not isomorphic. For the case when $d\leq n-2$  we know from~\cite{Cox1} that Walled Brauer Algebra is no longer semisimple. 
In this part we also give sketch of irreps construction which is though more complicated (see Subsection~\ref{mniejsze}). One can see that algebra $\mathcal{A}$ is an algebra of finite matrices. From~\cite{Landsman} we know that every finite $*$ algebra is semisimple, so this proves that algebra $\mathcal{A}$ in contrast to Walled Brauer Algebra is semisimple for an arbitrary dimension $d$ of Hilbert space and number of subsystems $n$.

\subsection{Some concepts and notation regarding symmetric group}
\label{some}
The Problem of partially transposed permutation operators  in a natural way imposes some structures on symmetric group.  Here we will introduce some notation, whose significance will become clear in the next paragraph, where we will relate it to  algebra of partially transposed permutation operators.
\begin{notation}
\label{not1}
Any permutation $\sigma \in S(n)$ defines, in a natural and unique way, two
natural numbers $a,b\in \{1,2,...,n\}$%
\[
n=\sigma(a), \qquad b=\sigma (n).
\]
Thus we may characterize any permutation (see example~\ref{ex} in Appendix) by these two numbers in the
following way%
\[
\sigma \equiv \sigma _{(a,b)}. 
\]%
Note that in general $a,b$ may be different except the case, when one of
them is equal to $n,$ because in this case we have%
\[
a=n\Leftrightarrow b=n. 
\]%
When $a=n=b,$ then $\sigma (n)=n$ and we have $\sigma =\sigma _{(n,n)}\equiv
\sigma _{n}.$ 
\end{notation}
For the further reasons let us define family of maps $f_{ab}:S(n-2) \rightarrow S(n)$ in the way that for fixed indeces  $a,b$ we have
\be
\label{eq:maps}
S(n-2) \ni \sigma  \mathop{\longmapsto}\limits^{f_{ab}} \sigma_{ab}:= \pi_b \circ \sigma \circ (nn-1) \circ \pi_a^{-1} \in S(n),
\ee 
where $\pi_k\in S(n-1)$. For a while we will not determine their particular form. Only in Section~\ref{construct} we will take $\pi_k = (n-1 k)$. Let us note that  maps $f_{ab}$ are invertible on their images, so for given $\sigma_{ab}\in S(n)$ we can define $f^{-1}_{ab}(\sigma_{ab})=\sigma$.  For reasons that will become clear in the proof of Lemma~\ref{actionOnvij} we define also another family of maps i.e. $f_c: S(n-1)\rightarrow S(n-2)$, where $1\leq c \leq n-1$ in a following way:
\be
\label{eq:maps2}
S(n-1) \ni \sigma \mathop{\longmapsto}\limits^{f_{c}} \sigma_c:=\pi_{\sigma[c]}^{-1}\circ \sigma \circ \pi_c\in S(n-2).
\ee
The mappings $f_{ab}$ divide  group $S(n)$ into subsets $S_{ab}$, such that:
\be
\label{division}
S(n)=\bigcup_{a,b=1}^{n-1}S_{ab}\cup S(n-1) \qquad \text{and} \qquad S_{ab} \cap S_{cd}=\emptyset \quad \text{for} \quad ab \neq cd.
\ee
Division into classes $S_{ab}$ for $n=3,4$ and $5$ is presented in Appendix~\ref{class}. One can see that we have isomorphism
$S_{ab}\cong S(n-2)$ which is defined for fixed $a,b$ by map $f_{ab}$ from equation~\eqref{eq:maps}.

\subsection{The algebra of partially transposed permutation operators}
\label{alg}
We can start our main considerations from following observation. Consider permutation operator $\opV(\sigma_{ab})$, where $\sigma_{ab}\in S(n)$
\be
\opV(\sigma_{ab}):=\sum_{i_1,\ldots, i_n}|i_{\sigma^{-1}(1)},\ldots,i_{\sigma^{-1}(b)=n},\ldots, i_{\sigma^{-1}(n)=a}\>\<i_{1},\ldots,i_a,\ldots,i_{n}|.
\ee
After partial transposition on last subsystem which we will denote by $'$ and whenever $\sigma_{ab}(n)\neq n$, $\sigma_{ab}(n)=a$ and $\sigma^{-1}_{ab}[n]=b$ we obtain
\be
\label{partialy}
\begin{split}
\opV'(\sigma_{ab})&=\sum_{i_1,\ldots, i_n} |i_{\sigma^{-1}(1)},\ldots, i_{\sigma^{-1}(b)=n},\ldots,i_{n}\>\<i_1,\ldots,i_{a},\ldots,i_{\sigma^{-1}(n)=a}|=\sum_{\phi}|\phi\>|\widetilde{\Phi}^+\>_{bn}\<\phi'(\phi)|\<\widetilde{\Phi}^+|_{an},
\end{split}
\ee
where $|\phi\>, |\phi'(\phi)\>$ are some vectors defined on $n-2$ subsystems. $|\widetilde{\Phi}^+\>_{an}$, $|\widetilde{\Phi}^+\>_{bn}$ are unnormalized maximally entangled state between last subsystem and $a^{\text{th}}$ and $b^{\text{th}}$ respectively. So it is clear that any partially transposed operator $\opV'(\sigma_{ab})\in \mathcal{A}$ is built from maximally entangled state between last subsystem and subsystems $a^{\text{th}}$ or $b^{\text{th}}$  and vectors of length $n-2$ whenever $\sigma(n)\neq n$. We see that operators $\opV'(\sigma_{ab})$ vanish outside space spanned by vectors of the form $|\phi\>|\widetilde{\Phi}^+\>_{kn}$, where $1 \leq k \leq n-1$ and vector $|\phi\>$ is on $n-2$ subsystems except subsystems $k^{\text{th}}$ and $n^{\text{th}}$. We will denote this space by $\mathcal{H}_{\mathcal{M}}$ (see Appendix~\ref{FullDescription} to detailed discussion). It is easy to see that space $\mathcal{H}_{\mathcal{M}}$ is invariant under action of elements from the group $S(n-1)$, because such elements do not affect last subsystem in equation~\eqref{partialy}. So space $\mathcal{H}_{\mathcal{M}}$ is invariant subspace for full algebra $\mathcal{A}$.
We can rewrite definitions of maps~\eqref{eq:maps},~\eqref{eq:maps2} in terms of operators  mapping. We start from definition of family of functions, which maps elements from algebra $\mathbb{C}[S(n-2)]$ onto elements from algebra $\mathcal{A}$ and second family of functions which map algebra $\mathbb{C}[S(n-1)]$ onto $\mathbb{C}[S(n-2)]$, i.e. we have following
\begin{definition}
\label{Fab}
Let us define set of maps $\mathcal{F}_{ab}^t: \mathbb{C}[S(n-2)] \rightarrow \mathcal{A}$, which acts on arbitrary operator $\operatorname{X} \in \mathbb{C}[S(n-2)]$ according to formula:
\be
\label{eq:supop}
\mathcal{F}_{ab}^t(\operatorname{X}):=\Pi_b \operatorname{X} \opV'(nn-1)\Pi_a^{-1},
\ee
where operators $\Pi_a$ and $\Pi_b^{-1}$ represent permutations $\pi_a, \pi_b^{-1}$ from equation~\eqref{eq:maps} and permutation operator $\opV(nn-1)$ represents  cycle $(nn-1)$. 
Let us define also set of maps $\mathcal{F}_c:\mathbb{C}[S(n-1)] \rightarrow \mathbb{C}[S(n-2)]$, which act on arbitrary operator  $\operatorname{X}\in \mathbb{C}[S(n-1)]$ according to formula:
\be
\label{eq:supop2}
\mathcal{F}_{c}(\operatorname{X}):=\Pi_{\sigma[c]}^{-1}\opV(\sigma)\Pi_c,
\ee
where operators $\Pi_{\sigma[c]}^{-1}$ and $\Pi_c$ represent permutations $\pi_{\sigma[c]},\pi_c^{-1}$ from equation~\eqref{eq:maps2}.
\end{definition}
Importance of this definition we will see if we consider action of elements from algebra $\mathbb{C}[S(n-1)]$ onto some operators $v_{ij}^{ab}(\alpha)$, which are defined later.

\begin{remark}
\label{r:rem1}
If we consider representation of the algebra $\mathcal{A}$ on Hilbert space $\mathcal{H}^{\ot n}$, then  maps $\mathcal{F}_{ab}^t$ can be lifted to a linear mapping that maps:
\be
\label{eq:cos}
\mathcal{F}_{ab}^t:\mathcal{L}\left(\mathcal{H}^{\ot (n-1)}\right) \longrightarrow \mathcal{L}\left(\mathcal{H}^{\ot n}\right)\quad \text{of the form} \quad \mathcal{F}_{ab}^t(\operatorname{X}):=\operatorname{F}_b \operatorname{X} \operatorname{F}_a^{-1},
\ee
where $\operatorname{X}\in\mathcal{H}^{\ot (n-2)}$ and  operators $\operatorname{F}_k$ assign to the vector $|\phi \> \in \mathcal{H}^{\ot (n-2)}$ vector $|\psi_k\>=\sqrt{d}\opV(\pi_k)|\phi\>|\Phi_+\> \in \mathcal{H}^{\ot n}$. Vector $|\Phi_+\>=\frac{1}{\sqrt{d}}\sum_{l=1}^d|ll\>$ is maximally entangled state between subsystems $n$ and $n-1$ and operators $\opV(\pi_k)$ represent permutations $\pi_k\in S(n-1)$ for $1 \leq k \leq n-1$.
\end{remark}
\subsection{Set of basis vectors. Nonorthogonal operator basis.}
\label{nonort}
In this paragraph we present how to construct nonorthogonal bases which span  irreducible  subspaces $\mathcal{H}^{\alpha,r}_{\mathcal{M}}$ in Hilbert space $\mathcal{H}_{\mathcal{M}}$. We start from the following definition:
\begin{definition}
\label{vecs0}
Define vectors which belong to $\mathcal{H}^{\alpha,r}_{\mathcal{M}}\subset\mathcal{H}^{\ot n}$ as
\be
\label{n10}
|\psi_{i}^{k}(\alpha,r)\>=\sqrt{d}\opV(\pi_k)|\phi_i(\alpha, r) \>|\Phi_+\>,
\ee
where index $\alpha$ labels irreps, index $r$ labels their multiplicities and $d$ is dimension of local Hilbert space.
\end{definition}
Using set of vectors $|\psi_{i}^{k}(\alpha,r)\>$ we can easily define operators onto space of every irrep labeled by $\alpha$:

\def\definitv0{
Let us define operators $v_{ij}^{ab}(\alpha)$ which project  onto subspace of given representation $\alpha$ as follows:
\be
v_{ij}^{ab}(\alpha)=\sum_r|\psi_i^a(\alpha,r)\>\<\psi_j^b(\alpha,r)|,
\ee
where sum runs over multiplicity of representation $\alpha$.}

\begin{definition}\label{Definitv0}
\definitv0
\end{definition}

Now let us present some discussion about properties of vectors $|\psi_i^k(\alpha,r)\>$ given in Definition~\ref{vecs0}. We collect this in two following lemmas:
\def\lemmailsk{
For two vectors $|\psi_i^a(\alpha,r)\>, |\psi_j^b(\alpha,r)\>\in \mathcal{H}_{\mathcal{M}}^{\alpha, r}$  scalar product is given by
\be
\<\psi_i^a(\alpha,r)|\psi_j^b(\alpha,r)\>=\<\phi_j(\alpha,r)|\opX(\chi_{ab})|\phi_i(\alpha,r)\>,
\ee 
where $\opX(\chi_{ab})=\tr_{n,n-1}\left(\opV(n,n-1)\Pi_{a}^{-1}\Pi_b\right)$.
}

\begin{lemma}\label{Lemmailsk}
\lemmailsk
\end{lemma}
We would like to know explicit form of permutation $\chi_{ab}$ for various relations between indeces $a,b$. In next lemma answer is presented.
\def\lemmaprop{
Operator $\opX(\chi_{ab})$ has the following form:
\begin{equation}
\opX(\chi_{ab}) = \left\{ \begin{array}{ll}
d \cdot \text{\noindent
\(\mathds{1}\)}_{1\ldots n-2} &  , \  \text{for}  \  a=b\\
\opV(\chi_{ab}) &, \  \text{for} \   a \neq b,
\end{array} \right.
\end{equation}
where $\chi_{ab}=(nn-1)\circ \pi_a^{-1} \circ \pi_b \circ (n-1n\pi_b^{-1}[a])$ is permutation from $S(n-2)$. 
}

\begin{lemma}\label{Lemmaprop}
\lemmaprop
\end{lemma}

Operators $v_{ij}^{ab}(\alpha)$ satisfy the  following composition rule:
\def\lemmavij{
Operators $v_{ij}^{ab}(\alpha)$ from Definition~\ref{Definitv0} satisfy following law of composition
\be
v_{ij}^{ab}(\alpha )v_{kl}^{cd}(\beta)=\delta_{\alpha \beta}d^{\delta_{bc}}\varphi_{jk}^{\alpha}(\chi_{bc})v_{il}^{ad}(\alpha),
\ee
where  $\varphi_{jk}^{\alpha}(\chi_{bc})$ is representation element of some permutation $\chi_{bc}$ given in Lemma~\ref{Lemmaprop} (see Appendix~\ref{FullDescription}).
}

\begin{lemma}\label{Lemmavij}
\lemmavij
\end{lemma}

Now we can see directly the reason for introducing definitions of functions $f_{ab}$ and $f_a$ in equations~\eqref{eq:maps},~\eqref{eq:maps2} respectively. Namely we can easily define mapping of elements from algebra $\mathbb{C}[S(n-2)]$ onto elements from algebra $\mathcal{M}$:
\be
\mathbb{C}[S(n-2)]\ni \opV(\sigma) \mathop{\longmapsto}\limits^{\mathcal{F}_{ba}^t} \opV'(\sigma_{ab}) \in \mathcal{M}.
\ee
Because of linearity of map $\mathcal{F}_{ab}^t$ we can act also on linear combinations of elements from algebra $\mathbb{C}[S(n-2)]$, in particular on operators~\footnote{It follows directly from equation~\ref{vecs0} and definition  of operators $\opE_{ij}^{\alpha}$, i.e. $\opE_{ij}^{\alpha}=\sum_r |\phi_i(\alpha, r)\>\<\phi_j(\alpha,r)|$, where number $r$ is multiplicity of given irrep $\alpha$. See also Appendix~\ref{A2} where we present basic discussion about properties of operators $\opE_{ij}^{\alpha}$} $\opE_{ij}^{\alpha}$:
\be
\label{eq:con}
\mathbb{C}[S(n-2)]\ni \opE_{ij}^{\alpha} \mathop{\longmapsto}\limits^{\mathcal{F}_{ba}^t} v_{ij}^{ab}(\alpha) \in \mathcal{M}.
\ee
Note that, because of equation~\eqref{eq:supop}  transformations $\mathcal{F}_{ab}^t$ are invertible on their images.
\\
Now we are ready to formulate the main result of our work in the following  
\def\theoremForVe{
a) Operators $v_{ij}^{ab}(\alpha)$ can be written in terms of partially transposed permutation operators in the following way
\be
v_{ij}^{ab}(\alpha)=\frac{d_{\alpha}}{(n-2)!}\sum_{\sigma\in S(n-2)}\varphi_{ij}^{\alpha}(\sigma)\opV'(f_{ba}(\sigma)),
\label{eq:Comp1}
\ee
where $f_{ba}(\sigma)=\pi_a \circ \sigma \circ (nn-1) \circ \pi_b^{-1}\in S(n)$. Operators $v_{ij}^{ab}(\alpha)$ are elements of algebra $\mathcal{A}$.\\
b) Operators $\opV'(\sigma_{ab})$ can be written in terms of operators $v_{ij}^{ab}(\alpha)$ as follows
\be
\label{eq:decomp2}
\opV'(\sigma_{ab})=\bigoplus_{\alpha}\sum_{ij}\varphi_{ij}^{\alpha}(f_{ab}^{-1}(\sigma_{ab}))v_{ij}^{ba}(\alpha).
\ee
}

\begin{theorem} \label{ccomp1}
\theoremForVe
\end{theorem}
It is easy to see that relations
\be
\label{sim2}
\mathbb{C}[S(n-2)]\ni \opV(\sigma)=\bigoplus_{\alpha}\sum_{ij}\varphi^{\alpha}_{ij}(\sigma)\opE_{ij}^{\alpha}, \qquad \opV'(\sigma_{ab})=\bigoplus_{\alpha}\sum_{ij} \varphi_{ij}^{\alpha}(f_{ab}^{-1}(\sigma_{ab}))v_{ij}^{ba}(\alpha)\in \opV'[S(n)].
\ee
together with  equations from~\eqref{eq:maps} to~\eqref{eq:supop} imply directly Theorem~\ref{ccomp1}. Let us note that the theorem implies that our invariant subspaces $\mathcal{H}^{\alpha,r}_{\mathcal{M}}$ are irreducible. Namely, We know that operators $v_{ij}^{ab}(\alpha)$ are linearly independent and their number is equal to square of dimension, so they span full operator basis, so that indeed the representation is irreducible.

Using definition of maps $f_{ab}$ from equation~\eqref{eq:maps} we can write how elements from algebras $\mathcal{M}$ and $\mathbb{C}[S(n-1)]$ act on  operators $v_{ij}^{ab}(\alpha)$: 
\def\ActionOnvij{
\label{actiononv}
Operators $\opV'(\sigma_{ab})\in \mathcal{M}$ and  operators $\opV(\sigma)\in \mathbb{C}[S(n-2)]$ act on $v_{ij}^{cd}(\alpha)$ according to formulas:
\be
\begin{split}
\opV'(\sigma_{ab})v_{kl}^{cd}(\alpha)&=d^{\delta_{ac}}\sum_i\varphi_{ik}^{\alpha}\left(f_{ab}^{-1}(\sigma_{ab})\circ \chi_{ac}\right)v_{il}^{bd}(\alpha),\\
\opV(\sigma)v_{ij}^{ab}(\alpha)&=\sum_k\varphi_{ki}(f_a(\sigma))v_{kj}^{\sigma[a]b}(\alpha),
\end{split}
\ee
where permutation $\chi_{ac}$ is given in Lemma~\ref{Lemmaprop} (see Appendix~\ref{FullDescription}), $f_{ab}^{-1}(\sigma_{ab})$ is inversion of $f_{ab}(\sigma)$ given in equation~\eqref{eq:maps} and $f_a(\sigma)$ is given in equation~\eqref{eq:maps2}.
}

\begin{lemma}\label{actionOnvij}
\ActionOnvij
\end{lemma}
\subsection{Construction of biorthonormal basis. Irreducible representations for $\opV'(\sigma)$.}
\label{bioort}
Let us emphasize, that operators $v_{ij}^{ab}(\alpha)$ are not orthonormal in Hilbert-Schmidt norm. The lack of this property, we can remove by redefining set of $v_{ij}^{ab}(\alpha)$, or in other words by defining new set of operators.
We start from following crucial definition:
\begin{definition}
\label{Q}
For any irreducible representation $\varphi^{\alpha}$ of the group $S(n-2)$ we define the block matrix
\be
\label{block}
Q_{ij}^{ab}(\alpha)=d^{\delta_{ab}}\varphi_{ij}^{\alpha}(\chi_{ab}), \quad \text{for} \quad 1\leq a,b \leq n-1, \quad  \quad 1\leq a,b \leq \operatorname{dim}\varphi^{\alpha},
\ee
permutation $\chi_{ab}\in S(n-2)$ is given in Lemma~\ref{Lemmaprop} and for the case when $a=b$ we define $\chi_{ab}=e$, where $e$ is the identity component of the group $S(n-2)$. The blocks of the matrix $Q_{ij}^{ab}(\alpha)$ are labelled by pair of indeces $(a,b)$ whereas the elements of the blocks are labelled by indeces of irreducible representation $\varphi^{\alpha}=\varphi_{ij}^{\alpha}(\chi_{ab})$ of the group $S(n-2)$, $Q_{ij}^{ab}(\alpha)\in M((n-1)m_{\alpha},\mathbb{C})$ where $m_{\alpha}=\operatorname{dim}\varphi^{\alpha}$.
\end{definition}
\begin{remark}
Above defined matrix is nothing else but Gram matrix of the basis $\{|\psi_i^k(\alpha,r)\>\}$ like in right hand side of equation~(24) in Lemma~\ref{Lemmailsk}. In further considerations, more specifically in Section~\ref{construct} we calculate matrix elements of irreps for specific choice of permutations $\pi_a$ and $\pi_b$. Namely we take the simplest one - transpositions, i.e. $\pi_a=(an-1)$ and $\pi_b=(bn-1)$.
Using general form of permutation $\chi_{ab}$ form Lemma~\ref{Lemmaprop}  we can easily show that 
\be
\label{specific}
\opX(\chi_{ab}) = \left\{ \begin{array}{ll}
d \cdot \text{\noindent
 \(\mathds{1}\)}_{1\ldots n-2} &  , \  \text{for}  \  a=b\\
\text{\noindent
\(\mathds{1}\)}_{1\ldots n-2} &  , \  \text{for}  \  a \ \text{or} \ b=n-1,\\
\opV(ab) &, \  \text{for} \   a \neq b,
\end{array} \right.
\ee
Finally in this case matrix $Q_{ij}^{ab}(\alpha)$ has form:
\be
\label{formQ}
Q_{ij}^{ab}(\alpha)=\begin{pmatrix}d\text{\noindent
\(\mathds{1}\)} & \varphi^{\alpha}(12) & \ldots & \varphi^{\alpha}(1n-2) & \text{\noindent
\(\mathds{1}\)} \\ \varphi^{\alpha}(21) & d\text{\noindent
\(\mathds{1}\)} & \ldots & \varphi^{\alpha}(2n-2) & \text{\noindent
\(\mathds{1}\)} \\ \vdots & &\ddots & & \vdots\\ \text{\noindent
\(\mathds{1}\)} & & \ldots & & d\text{\noindent
\(\mathds{1}\)} \end{pmatrix},
\ee
where every $\varphi^{\alpha}(ab)=\{\varphi_{ij}^{\alpha}(ab)\}$ is a representation matrix of permutation $(ab)$ in irrep of $S(n-2)$ labelled by $\alpha$. It is worth to mention here that in general case there is always possibility to chose matrices $\varphi^{\alpha}$ to be unitary, so we get $\varphi^{\alpha}_{ij}(ab)=\bar{\varphi}^{\alpha}_{ji}(ab)$. In our paper our constrains are even stronger because representations $\varphi^{\alpha}(ab)$ are in the form of symmetric and real matrices, so we have $\varphi^{\alpha}_{ij}(ab)=\varphi^{\alpha}_{ji}(ab)$. Because matrix $Q_{ij}^{ab}(\alpha)$ comes from set of vectors it can happen that if some of vectors are linearly dependent Gram matrix is no longer strictly positive. 
\end{remark}
Positivity property of matrix $Q_{ij}^{ab}(\alpha)$ is strictly connected with relation between dimension of local Hilbert space $d$ and number of subsystems $n$. This dependence we present it in the following
\def\theoremForMatrix{
Suppose that the representation $\varphi^{\alpha}$ of the group $S(n-2)$ in the matrix $Q_{ij}^{ab}(\alpha)\in M((n-1)m_{\alpha},\mathbb{C})$ is unitary (but not necessarily irreducible).
Under this simple condition we have: if $d > n - 2$ then the matrix $Q_{ij}^{ab}(\alpha)$ is
(strictly) positive i.e. $Q_{ij}^{ab}(\alpha) > 0$ and consequently if $d > n - 2$ the matrix $Q_{ij}^{ab}(\alpha)$
is invertible. These statements are the consequence of the following inequality
\be
x^+Q_{ij}^{ab}(\alpha )x \geq (d-n+2)\sum_{i=1}^{n-1}||x_i||^2,
\ee
where
\be
x^+=(x_1^+,x_2^+,\ldots ,x_{n-1}^+)\in \mathbb{C}^{(n-1)m_{\alpha}}
\ee
is the block vector in $\mathbb{C}^{(n-1)m_{\alpha}}$, $x_i^+\in \mathbb{C}^{m_{\alpha}}$ and $||x_i||$ is the standard norm of the vector $x_i \in \mathbb{C}^{m_{\alpha}}$.
}

\begin{theorem}\label{theoremForMatrix}
\theoremForMatrix
\end{theorem}

Using Definition~\ref{Q} and Theorem~\ref{theoremForMatrix} we can construct new operator basis which is orthonormal in Hilbert-Schmidt norm in the regime $d>n-2$.
\def\definitdeff0{
Let us define new operators $\omega_{ij}^{ab}(\alpha)$ which are connected with operators $v_{ij}^{ab}(\alpha)$ by following transformation rule:
\be
\label{deff0}
\omega_{ij}^{ab}(\alpha)=\sum_{kc}D_{jc}^{bk}(\alpha)v_{ic}^{ak}(\alpha),
\ee
where $D_{jc}^{bk}(\alpha)=((Q^{-1})_{jc}^{bk}(\alpha))$.
}

\begin{definition}\label{Definitdeff0}
\definitdeff0
\end{definition}

One can see that operators $\omega_{ij}^{ab}(\alpha)$ from above definition have similar form to operators $v_{ij}^{ab}(\alpha)$ from Definition~\ref{Definitv0}. Namely we have that:
\be
\label{bi}
\omega_{ij}^{ab}(\alpha)=\sum_r|\psi_i^a(\alpha,r)\>\<\phi_j^b(\alpha,r)|, \quad \text{where} \quad |\phi_j^b(\alpha)\>=\sum_{kc}D_{cj}^{kb}(\alpha)|\psi_c^k(\alpha)\>
\ee
and system $\{|\psi_i^a(\alpha)\>, |\phi_j^b(\alpha)\>\}$ form biorthogonal basis for any fixed $r$.\\
Now we are ready to show that new operators $\omega_{ij}^{ab}(\alpha)$ have required property of composition:
\def\compOfOmegas{
Operators $\omega^{ab}_{ij}(\alpha)$ satisfy the following composition rule
\be
\omega_{ij}^{ab}(\alpha)\omega_{kl}^{cd}(\beta)=\delta_{\alpha \beta}\delta^{bc}\delta_{jk}\omega_{il}^{ad}(\alpha).
\ee
}

\begin{lemma}\label{CompOfOmegas}
\compOfOmegas
\end{lemma}

Since operators $\omega_{ij}^{ab}(\alpha)$ are linear combinations of operators $v_{ij}^{ab}(\alpha)$, they
satisfy the same rules of transformations like operators $v_{ij}^{ab}(\alpha)$ (see Appendix~\ref{A1}):
\def\ActionOnF{
Operators $\opV'(\sigma_{ab})\in \mathcal{M}$ and  operators $\opV(\sigma)\in \mathbb{C}[S(n-1)]$ act on $\omega_{ij}^{cd}(\alpha)$ according to formulas:
\be
\begin{split}
\opV'(\sigma_{ab})\omega_{kl}^{cd}(\alpha)&=d^{\delta{ac}}\sum_i \varphi_{ik}^{\alpha}(  f_{ab}^{-1}(\sigma_{ab})\circ \chi_{ac})\omega_{il}^{bd}(\alpha),\\
\opV(\sigma)\omega_{ij}^{ab}(\alpha)&=\sum_k\varphi_{ki}^{\alpha}(f_a(\sigma))\omega^{\sigma[a]b}_{kj}(\alpha).
\end{split}
\ee
where $f_{ab}^{-1}(\sigma_{ab})$ is inversion of $f_{ab}(\sigma)$ given in equation~\eqref{eq:maps} and $f_a(\sigma)$ is given in equation~\eqref{eq:maps2}.
}

\begin{proposition}\label{actionOnF}
\ActionOnF
\end{proposition}

At the end of this section we can say that because of Definition~\ref{Definitdeff0} (we assume that gram matrix $G$ is nonsingular), operators $\omega^{ab}_{ij}(\alpha)$ are also elements of algebra $\mathcal{M}$. More precisely we can reformulate Theorem~\ref{ccomp1}:
\def\theoremForOmegas{
a) Operators $\omega_{ij}^{ab}(\alpha)$ can be written in terms of partially transposed permutation operators in the following way
\be
\omega_{ij}^{ab}(\alpha)=\frac{d_{\alpha}}{(n-2)!}\sum_{\sigma \in S(n-2)}\sum_{c}\varphi_{ij}^{\alpha}(\sigma \circ \chi_{bc}^{-1})\opV'(f_{ca}(\sigma)),
\label{eq:CComp1}
\ee
where $f_{ca}(\sigma)=\pi_a \circ \sigma \circ (nn-1) \circ \pi_c^{-1}\in S(n)$ and operators $\omega_{ij}^{ab}(\alpha)$ are elements of algebra $\mathcal{M}$.\\
b) Operators $\opV'(\sigma_{ab})$ can be written in a form
\be
\label{eq:ddecomp2}
\opV'(\sigma_{ab})=\bigoplus_{\alpha}\sum_{cik}\varphi_{ik}^{\alpha}(\sigma \circ \chi_{ac})\omega_{ik}^{bc}(\alpha),
\ee
where operators $\omega_{ij}^{ab}(\alpha)$ are given by Definition~\ref{Definitdeff0}.
}
\begin{theorem} \label{thForOmegas}
\theoremForOmegas
\end{theorem}
Now we are ready to give formulas for matrix elements of an arbitrary operator $X$ in biorthogonal basis $\{|\psi_i^a(\alpha)\>, |\phi_j^b(\alpha)\>\}$ (see also Appendix~\ref{A3}). In this section we give  explicit formulas for matrix elements of partially transposed permutation operators $\opV'(\sigma_{ab})\in \mathcal{M}$ and $\opV(\sigma)\in \mathbb{C}[S(n-1)]$ in above-mentioned basis:
\def\irreps{
Matrix elements of irreducible representations of permutation operators in biorthogonal basis $\{|\psi_i^a(\alpha)\>, |\phi_j^b(\alpha)\>\}$ are given by following formulas:
\be
\begin{split}
[\widetilde{\opV}'_{\alpha}(\sigma_{ab})]_{dl,ck}&=d^{\delta_{ac}}\varphi_{lk}^{\alpha}(f^{-1}_{ab}(\sigma_{ab})\circ \chi_{ac})\delta_{bd},\\
[\widetilde{\opV}_{\alpha}(\sigma )]_{dl,ck}&=\varphi_{lk}^{\alpha}(f_c(\sigma))\delta_{\sigma[c]d},
\end{split}
\ee
where $\varphi_{lk}^{\alpha}(\cdot)$ is matrix elements of permutation operator for permutation from $S(n-2)$ for given irrep $\alpha$.
}

\begin{lemma}\label{ElementsIrreps}
\irreps
\end{lemma}

\section{Construction of irreducible representations. Examples of matrix elements}
\label{construct}
This section is split into three main parts. In the first part we focus on the case when $d>n-2$, namely we present matrices of irreps from algebra $\mathcal{M}$ for small numbers of $n$. In the second part we give discussion when condition $d>n-2$ is not fulfilled. For this case we know from general theory~\cite{Cox1} that Walled Brauer Algebra is no longer semisimple. Example that algebra $\mathcal{M}$, even for $d\leq n-2$ is still semisimple is presented. At the end of second part we outline the new problem corresponding with our method of irreps construction for this case. We ask is it possible to say something general about rank of Gram Matrix $Q_{ij}^{ab}(\alpha)$ from Definition~\ref{Q}, when some basis vectors $|\psi_i^a(\alpha,r)\>$ are linearly dependent. Finally in the third part of this section we show how to calculate multiplicities of irreps.
\subsection{Construction of irreducible representations when $d>n-2$. Examples of irreducible representations for $n=3,4,5$.}
\label{s:examples}
Here we give few explicit examples of irreps from algebra $\mathcal{M}$ for various values of $n$, i.e. $n=3,4$ and $5$. For our calculations we use exactly Lemma~\ref{ElementsIrreps}. To do so we have to choose explicit permutations $\pi_k$ from Definition~\ref{vecs0}. In this paragraph we use the simplest one - transposition between elements $k$ and $n-1$, so 
$\pi_k=(kn-1)$. Thanks to this formulas for embedding functions $f_{ab}$ and $f_a$ from equations~\eqref{eq:maps},~\eqref{eq:maps2} have a form
\be
f_{ab}(\sigma)=(bn)(bn-1)\circ \sigma \circ (an-1),\qquad f_a(\sigma)=(\sigma[a]n-1)\circ \sigma \circ (an-1).
\ee
To obtain explicit matrix elements from lemma~\ref{ElementsIrreps} we need also direct form of permutation $\chi_{ab}$ from  Lemma~\ref{Lemmaprop} (see Appnedix~\ref{FullDescription}) in our computational basis
\be
\label{specific2}
\opX(\chi_{ab}) = \left\{ \begin{array}{ll}
d \cdot \text{\noindent
 \(\mathds{1}\)}_{1\ldots n-2} &  , \  \text{for}  \  a=b\\
\text{\noindent
\(\mathds{1}\)}_{1\ldots n-2} &  , \  \text{for}  \  a \ \text{or} \ b=n-1,\\
\opV(ab) &, \  \text{for} \   a \neq b,
\end{array} \right.
\ee
 It is worth to mention here that for irreps which are labelled by symmetric or antisymmetric partitions from $S(n-2)$ equations for matrix elements take very simply form:
\be
\begin{split}
[\widetilde{\opV}'_{(n)}(\sigma_{ab})]_{d1,c1}&=d^{\delta_{ac}}\delta_{bd},\\
[\widetilde{\opV}_{(n)}(\sigma )]_{d1,c1}&=\delta_{\sigma[c]d},
\end{split}
\ee
for symmetric case~\footnote{In this case representations are one-dimensional and equal to 1. They do not depend on permutations $f^{-1}_{ab}(\sigma_{ab})\circ \chi_{ac}$. }, and
\be
\begin{split}
[\widetilde{\opV}'_{(1^n)}(\sigma_{ab})]_{d1,c1}&=d^{\delta_{ac}}\operatorname{sgn}(f^{-1}_{ab}(\sigma_{ab})\circ \chi_{ac})\delta_{bd},\\
[\widetilde{\opV}_{(1^n)}(\sigma )]_{d1,c1}&=\operatorname{sgn}(f_c(\sigma))\delta_{\sigma[c]d},
\end{split}
\ee
for antisymmetric case, where $\operatorname{sgn}(\cdot)$ is signum function. We know that any element from algebra of permutation operators can be obtained by proper combination of its generators $\opV(i,i+1)$, where $i=1,\ldots,n-1$. The same situation we have also for algebra $\mathcal{M}$ of partially transposed permutation operators. Namely to find any element from it, we have to construct proper combination of generators. In our case  set of generators for algebra $\mathcal{M}$ is equal to $\{\opV(12),\opV(23),\ldots,\opV'(n-1,n)=\Phi^+_{n-1,n}\}$.  In representative language we can rewrite them in a form $\widetilde{\opV}_{\alpha}(k,k+1)$, where $k=1,\ldots, n-1$. So to calculate representation of any element from algebra $\mathcal{M}$ is enough to calculate following representations of generators: $\widetilde{\opV}_{\alpha}'(k,k+1)$, where $k=1,\ldots,n-1$. Note that we have only one generator on which partial transposition acts non trivially, i.e. when $k=n-1$, so in this case we have that $a=b=n-1$, so always $f^{-1}_{ab}(\sigma_{ab})\circ \chi_{ac}=e$ for any $c=1,\ldots ,n-1$. Thanks to this equation for matrix elements for antisymmetric case takes very simply form
\be
\label{elel2}
\begin{split}
[\widetilde{\opV}'_{(1^n)}(\sigma_{n-1})]_{d1,c1}&=d^{\delta_{n-1c}}\delta_{n-1d}, 
\end{split}
\ee
where $c=1,\ldots ,n-1$ and $\sigma_{n-1}=\sigma_{n-1,n-1}$.\\
{\bf Case $n=3$} In this case in algebra $\mathcal{M}$ we have only one two-dimensional irreps labelled by partition $\alpha_1=(1)$. List of generators is as follows:
\be
\widetilde{\opV}'_{\alpha_1}(e)=\begin{bmatrix} 1 & 0 \\ 0 & 1 \end{bmatrix}, \  \widetilde{\opV}'_{\alpha_1}(12)=\begin{bmatrix} 0 & 1\\ 1 & 0\end{bmatrix}, \ \widetilde{\opV}'_{\alpha_1}(23)=\begin{bmatrix} 0 & 0  \\ 1 & d\end{bmatrix}.
\ee
{\bf Case $n=4$} In this case in algebra $\mathcal{M}$ we have two three-dimensional irreps labelled by partitions $\alpha_1=(2)$ and $\alpha_2=(1,1)$. 
\begin{itemize}
\item For partition $\alpha_1=(2)$ we have following list of irreducible representations for generators:
\begin{equation}
\begin{split}
\widetilde{\opV}'_{\alpha_1}(e)=\begin{bmatrix}1 & 0 & 0\\ 0 & 1 & 0\\ 0 & 0 & 1 \end{bmatrix}, \ \widetilde{\opV}'_{\alpha_1}(12)=\begin{bmatrix}0 & 1 & 0\\ 1 & 0 & 0\\ 0 & 0 & 1\end{bmatrix}, \ \widetilde{\opV}'_{\alpha_1}(23)=\begin{bmatrix}1 & 0 & 0\\ 0 & 0 & 1\\ 0 & 1 & 0 \end{bmatrix}, \ \widetilde{\opV}'_{\alpha_1}(34)=\begin{bmatrix} 0 & 0 & 0\\0 & 0 & 0 \\1 & 1 & d\end{bmatrix}.
\end{split}
\end{equation}
\item For partition $\alpha_2=(1,1)$ we have following list of irreducible representations for generators:\\
\begin{equation}
\begin{split}
\widetilde{\opV}'_{\alpha_2}(e)=\begin{bmatrix}1 & 0 & 0\\ 0 & 1 & 0\\ 0 & 0 & 1 \end{bmatrix}, \ \widetilde{\opV}'_{\alpha_2}(12)=\begin{bmatrix}0 & -1 & 0\\ -1 & 0 & 0\\ 0 & 0 & -1\end{bmatrix}, \ \widetilde{\opV}'_{\alpha_2}(23)=\begin{bmatrix}-1 & 0 & 0\\ 0 & 0 & 1\\ 0 & 1 & 0 \end{bmatrix}, \ \widetilde{\opV}'_{\alpha_2}(34)=\begin{bmatrix} 0 & 0 & 0\\0 & 0 & 0 \\1 & 1 & d\end{bmatrix}.
\end{split}
\end{equation}
 \end{itemize}
{\bf Case $n=5$} In this case in algebra $\mathcal{M}$ we have three three-dimensional irreps labelled by partitions $\alpha_1=(3)$, $\alpha_2=(2,1)$ and $\alpha_3=(1,1,1)$.
\begin{itemize}
\item For partition $\alpha_1=(3)$ we have following list of irreducible representations for generators:
\be
\begin{split}
 &\widetilde{\opV}'_{\alpha_1}(e)=\begin{bmatrix}1 & 0 & 0 & 0\\ 0 & 1 & 0 & 0\\0 & 0 & 1 & 0 \\ 0 & 0 & 0 & 1\end{bmatrix}, \ \widetilde{\opV}'_{\alpha_1}(12)=\begin{bmatrix}0 & 1 & 0 & 0\\ 1 & 0 & 0 & 0\\0 & 0 & 1 & 0 \\ 0 & 0 & 0 & 1\end{bmatrix}, \ \widetilde{\opV}'_{\alpha_1}(23)=\begin{bmatrix}1 & 0 & 0 & 0\\ 0 & 0 & 1 & 0\\0 & 1 & 0 & 0 \\ 0 & 0 & 0 & 1\end{bmatrix},\\
 &\widetilde{\opV}'_{\alpha_1}(34)=\begin{bmatrix}1 & 0 & 0 & 0\\ 0 & 1 & 0 & 0\\0 & 0 & 0 & 1 \\ 0 & 0 & 1 & 0\end{bmatrix}, \ \widetilde{\opV}'_{\alpha_1}(45)=\begin{bmatrix}0 & 0 & 0 & 0\\ 0 & 0 & 0 & 0\\0 & 0 & 0 & 0 \\ 1 & 1 & 1 & d\end{bmatrix}.
\end{split}
\ee
\item For partition $\alpha_1=(2,1)$ we have following list of irreducible representations for generators:
\be
\begin{split}
 &\widetilde{\opV}'_{\alpha_1}(e)=\begin{bmatrix}1 & 0 & 0 & 0\\ 0 & 1 & 0 & 0\\0 & 0 & 1 & 0 \\ 0 & 0 & 0 & 1\end{bmatrix}, \ \widetilde{\opV}'_{\alpha_1}(12)=\begin{bmatrix}0 & 0 & 0 & 0\\ 0 & 0 & 0 & 0\\0 & 0 & 1 & 0 \\ 0 & 0 & 0 & -1\end{bmatrix}, \ \widetilde{\opV}'_{\alpha_1}(23)=\begin{bmatrix}-\frac{1}{2} & 0 & 0 & 0\\ 0 & 0 & \frac{\sqrt{3}}{2} & 0\\0 & \frac{\sqrt{3}}{2} & 0 & 0 \\ 0 & 0 & 0 & \frac{1}{2}\end{bmatrix},\\
 &\widetilde{\opV}'_{\alpha_1}(34)=\begin{bmatrix}-\frac{1}{2} & 0 & 0 & 0\\ 0 & \frac{1}{2} & 0 & 0\\0 & 0 & 0 & 0 \\ 0 & 0 & 0 & 0\end{bmatrix}, \ \widetilde{\opV}'_{\alpha_1}(45)=\begin{bmatrix}0 & 0 & 0 & 0\\ 0 & 0 & 0 & 0\\0 & 0 & 0 & 0 \\ 0 & 1 & 0 & d\end{bmatrix}.
\end{split}
\ee
\item For partition $\alpha_1=(1,1,1)$ we have following list of irreducible representations for generators:
\be
\begin{split}
 &\widetilde{\opV}'_{\alpha_1}(e)=\begin{bmatrix}1 & 0 & 0 & 0\\ 0 & 1 & 0 & 0\\0 & 0 & 1 & 0 \\ 0 & 0 & 0 & 1\end{bmatrix}, \ \widetilde{\opV}'_{\alpha_1}(12)=\begin{bmatrix}0 & -1 & 0 & 0\\ -1 & 0 & 0 & 0\\0 & 0 & -1 & 0 \\ 0 & 0 & 0 & -1\end{bmatrix}, \ \widetilde{\opV}'_{\alpha_1}(23)=\begin{bmatrix}-1 & 0 & 0 & 0\\ 0 & 0 & -1 & 0\\0 & -1 & 0 & 0 \\ 0 & 0 & 0 & -1\end{bmatrix},\\
 &\widetilde{\opV}'_{\alpha_1}(34)=\begin{bmatrix}-1 & 0 & 0 & 0\\ 0 & -1 & 0 & 0\\0 & 0 & 0 & 1 \\ 0 & 0 & 1 & 0\end{bmatrix}, \ \widetilde{\opV}'_{\alpha_1}(45)=\begin{bmatrix}0 & 0 & 0 & 0\\ 0 & 0 & 0 & 0\\0 & 0 & 0 & 0 \\ 1 & 1 & 1 & d\end{bmatrix}.
\end{split}
\ee
\end{itemize}
\subsection{Construction of irreducible representations when $d\leq n-2$.}
\label{mniejsze}
In this paragraph we consider situation when local dimension of Hilbert space $d\leq n-2$. Then dimension of algebra $\mathcal{A}$ decreases in comparison with Walled Brauer Algebra. In this case  Walled Brauer Algebra~\cite{Brundan1} has $\operatorname{dim}B_{n-1,1}(d)=n!$ and dimension of our algebra is equal to the number of linearly independent permutation operators. It is clear that already for $d\leq n-1$ algebra $\mathcal{A}$ has smaller dimension than $n!$ and these two algebras are not isomorphic anymore.
Another observation follows from Theorem~\ref{theoremForMatrix}. Namely from this theorem we know that for $d\leq n-2$ Gram matrix from Definition~\ref{Q} calculated on some partition $\alpha$ which runs over partitions of $n-2$ may not be invertible. In other words we see that Gram matrix  could not have full rank for some partition $\alpha$. Then it is clear that some of the basis vectors $|\psi_i^a(r,\alpha)\>$ from equation~\eqref{n10} are linearly dependent, so some formulas are not valid anymore. For example in this case can not use directly Definition~\ref{Definitdeff0} because matrix $D_{jc}^{bk}(\alpha)$ does not exist, so we do not have explicit formulas for irreps matrix elements like in Lemma~\ref{ElementsIrreps}. Here we give sketch of solution how we can omit this problem. \\
 Suppose that we choose from the set $\left\{|\psi_i^a(r,\alpha)\>\right\}$ of linearly dependent vectors a subset of linearly independent vectors which span new nonorthogonal basis. Suppose also that all basis elements are labelled by indices form the set $I=\left\{(i,a)\right\}$, where $a=1,\ldots,n-1$ and $i=1,\ldots,d_{\alpha}$. Then thanks to our previous considerations (or Appendix~\ref{FullDescription})  we can built operators using only vectors which are labelled by indices from set $I$:
\be
\sum_r |\psi^a_i(r,\alpha)\>\< \psi^b_j(r,\alpha)|=\text{\noindent
\(\mathds{1}\)}\ot |\psi_i^a(\alpha)\>\<\psi_j^b(\alpha)| =v_{ij}^{ab}(\alpha ),
\ee
where $(i,a)\in I$, and  $(j,b)\in I$. It is true that subset of operators
\be
\left\{v_{ij}^{ab}(\alpha) \ : \ (i,a)\in I, (j,b)\in I\right\}
\ee
spans a new operator basis for algebra $\mathcal{M}$.
\begin{example}
\label{eexx}
As an example we present here case when $n=4$, $d=2$ and we consider antisymmetric subspace of algebra $\mathcal{M}$ labelled by partition $\alpha=(1,1)$. We choose this partition because only in this case our Gram matrix does not have full rank (it has rank equal to two). We have three allowed vectors from which we can construct our basis:
\be
\label{dots}
|\psi_1\>=\frac{1}{\sqrt{2}}\left(|\cdot 10\cdot\>-|\cdot 01 \cdot\>\right), \ |\psi_2\>=\frac{1}{\sqrt{2}}\left(| 0 \cdot 1 \cdot\>-| 1\cdot 0 \cdot \>\right), \ |\psi_3\>=\frac{1}{\sqrt{2}}\left(|01 \cdot \cdot\>-|10 \cdot \cdot\>\right),
\ee
where $|\cdot \cdot\>=\sum_i|ii\>$. One can see that vectors from equation~\eqref{dots} are linearly dependent, namely we have that $|\psi_3\>=|\psi_2\>-|\psi_1\>$. Thanks to this we can assume that antisymmetric subspace i spanned by the set $\mathcal{M}=\{|\psi_1\>,|\psi_2\>\}$. Using vectors from the set $\mathcal{M}$ which are linearly independent we construct operator basis in our subspace, i.e. $v_{ab}=|\psi_a\>\<\psi_b|$~\footnote{For simplicity we keep here only indeces $a,b$.}, where $a,b=1,2$:
\be
\label{basis}
\begin{split}
v_{11}&=\begin{bmatrix}2 & 0\\ 1 & 0 \end{bmatrix}, \quad v_{12}=\begin{bmatrix}1 & 0\\ 2 & 0 \end{bmatrix}, \quad v_{21}=\begin{bmatrix}0 & 2\\ 0 & 1 \end{bmatrix}, \quad v_{22}=\begin{bmatrix}0 & 1\\ 0 & 2 \end{bmatrix}.
\end{split}
\ee
Rest of allowed basis operators $v_{3b}, v_{a3}$, for $a,b=1,2,3$ can be expressed in terms of operators from equation~\eqref{basis} because of linearly dependence between vectors $|\psi_a\>$, for $a=1,2,3$:
\be
\begin{split}
v_{13}&=v_{12}-v_{11}=\begin{bmatrix}-1 & 0\\ 1 & 0 \end{bmatrix},\quad v_{31}=v_{21}-v_{11}=\begin{bmatrix}-2 & 2\\ -1 & 1 \end{bmatrix},\quad v_{23}=v_{22}-v_{21}=\begin{bmatrix}0 & -1\\ 0 & 1 \end{bmatrix},\\
v_{32}&=v_{22}-v_{12}=\begin{bmatrix}-1 & 1\\ -2 & 2 \end{bmatrix},\quad v_{33}=v_{11}+v_{22}-v_{12}-v_{21}=\begin{bmatrix}1 & -1\\ -1 & 1 \end{bmatrix}.
\end{split}
\ee
Our next goal is to find matrix elements of partially transposed permutation operators $V'(\sigma)$, where $\sigma\in S(4)$.
To do this we have to find action of every operator $V'(\sigma)$ on basis vectors $|\psi_1\>,|\psi_2\>$, i.e. we have to find coefficient $a_{ij}$ in the following linear combination $\opV'(\sigma)|\psi_i\>=\sum_j a_{ij}|\psi_j\>$, where $i,j=1,2$. Thanks to this for permutation from group $S(3)\subset S(4)$ we have:
\be
\begin{split}
\opV(e)&=\frac{1}{3}\left(2v_{11}+2v_{22}-v_{12}-v_{21}\right)= \begin{bmatrix}1 & 0 \\ 0 & 1 \end{bmatrix}, \quad \opV(12)=\frac{1}{3}\left(2v_{12}+2v_{21}-v_{11}-v_{22}\right)=\begin{bmatrix}0 & 1 \\ 1 & 0 \end{bmatrix}\\
\opV(13)&=\frac{1}{3}\left(2v_{11}-v_{12}-v_{21}-v_{22}\right)=\begin{bmatrix}1 & -1 \\ 0 & -1 \end{bmatrix}, \quad \opV(23)=\frac{1}{3}\left(2v_{22}-v_{12}-v_{21}-v_{11}\right)=\begin{bmatrix}-1 & 0 \\ -1 & 1 \end{bmatrix},\\
\opV(123)&=\frac{1}{3}\left(2v_{21}-v_{11}-v_{12}-v_{22}\right)=\begin{bmatrix}-1 & 1 \\ -1 & 0 \end{bmatrix},\quad \opV(132)=\frac{1}{3}\left(2v_{12}-v_{11}-v_{21}-v_{22}\right)=\begin{bmatrix}0 & -1 \\ 1 & -1 \end{bmatrix}.
\end{split}
\ee
Now permutation operators for which partial transposition acts non-trivially:
\be
\begin{split}
&\opV'(14)=v_{11}, \quad \opV'(24)=v_{22}, \quad \opV'(23)(14)=-v_{11}, \quad \opV'(13)(24)=-v_{22},\\
&\opV'(124)=v_{21}, \quad \opV'(243)=v_{23}, \quad \opV'(1243)=-v_{23},\quad \opV'(1324)=-v_{21},\\
&\opV'(1423)=-v_{11}.
\end{split}
\ee
\be
\begin{split}
\opV'(34)&=v_{11}+v_{22}-v_{12}-v_{21}=v_{33}=\begin{bmatrix}1 & -1\\ -1 & 1\end{bmatrix},\quad \opV'(12)(34)=v_{12}+v_{21}-v_{11}-v_{22}=-v_{33}=\begin{bmatrix}-1 & 1\\ 1 & -1\end{bmatrix},\\
\opV'(143)&=-v_{11}-v_{12}=-v_{13}=\begin{bmatrix}1 & 0\\ -1 & 0\end{bmatrix},\quad \opV'(234)=v_{22}-v_{12}=v_{32}=\begin{bmatrix}-1 & 1\\ -2 & 2\end{bmatrix},\\
\opV'(134)&=v_{11}-v_{21}=-v_{31}=\begin{bmatrix}2 & -2\\1 & -1\end{bmatrix},\quad \opV'(1234)=v_{21}-v_{11}=v_{31}=\begin{bmatrix}-2 & 2\\-1 &1\end{bmatrix},\\
\opV'(1342)&=v_{12}-v_{22}=-v_{32}=\begin{bmatrix}1 & -1\\2 &-2\end{bmatrix},\quad \opV'(1432)=v_{12}-v_{11}=v_{13}=\begin{bmatrix}-1 & 0\\1 & 0\end{bmatrix}.
\end{split}
\ee
Summarizing we have shown that every operator from antisymmetric subspace can be written in terms of four operator $v_{ab}(\alpha)$, where $a,b=1,2$, so set $\{v_{ab}(\alpha)\}_{a,b=1}^2$ span required basis. 
\end{example}
One can see that we have here new important question. Namely we ask how rank of Gram matrices on every partition in algebra $\mathcal{M}$ depends on relation between local dimension of Hilbert space and number of subsystems. Full solution of this problem has considerable practical significance. We know that for $d> n-2$ matrices $Q_{ij}^{ab}(\alpha)$ have always full rank (see Theorem~\ref{theoremForMatrix}), so they are always invertible. Thanks to this and our construction described in Section~\ref{OurResults} we know that algebra $\mathcal{M}$ is semisimple.
For the case when $d \leq n-2$ from papers~\cite{Cox1,Brundan1} we know that Walled Brauer Algebra is no longer semisimple, but our algebra $\mathcal{M}$ still is (see Example~\ref{eexx}). Unfortunately matrices $Q_{ij}^{ab}(\alpha)$ for various $\alpha$ may not have full rank - so they are irreversible in general. To deal this problem we have to know how to choose from our linearly dependent set of vectors $\{|\psi_i^a(\alpha,r)\>\}$ new set which is linearly independent. Then we can construct new  Gram matrices $\widetilde{Q}_{ij}^{ab}(\alpha)$ with smaller dimensions and after that construct appropriate irreps using the same argumentation like in invertible case.  We have some partial results towards general resolution of this problem, which will be published elsewhere.

\subsection{Multiplicities of Irreps}
\label{multi}
In previous sections we have shown how to decompose algebra of partially transposed operators $\mathcal{A}$ into irreducible representations but we did not mention about their multiplicities. In this paragraph we will give short discussion related to this problem. First we recall here that in full Hilbert space $\mathcal{H}^{\ot n}$ we can distinguish  two subspaces $\mathcal{H}_{\mathcal{M}}$ and orthogonal to it $\mathcal{H}_{\mathcal{M}}^{\perp}=\mathcal{H}_{\mathcal{N}}$. 
Namely we have that $\mathcal{H}^{\ot n}=\mathcal{H}_{\mathcal{M}} \oplus \mathcal{H}_{\mathcal{N}}$. Thanks to this splitting we have two kinds of irreps. In subspace $\mathcal{H}_{\mathcal{M}}$ irreducible representations are labelled by partitions $\alpha_{\mathcal{M}} \vdash n-2$. 
\begin{corollary}
From Definition~\ref{vecs0} it follows directly that multiplicity of given irrep in $\mathcal{H}_{\mathcal{M}}$ is equal to multiplicity of irrep $\alpha$ in $S(n-2)$ which is known from the classical combinatorial rules~\cite{RWallach}. In subspace $\mathcal{H}_{\mathcal{N}}$ irreps are labelled by partitions $\alpha_{\mathcal{N}} \vdash n-1$.
\end{corollary}
From representation theory for symmetric group we know that $S(n)$ is centralizer of $U(d)$, so we have well known Schur-Weyl duality~\cite{Audenaert2006-notes},~\cite{RWallach} for permutation operators $\opV(\sigma)$, like in equation~\eqref{eq:swapdec}. For partially transposed permutation operators  we have quite similar situation~\cite{Kimura1}:
\be
\label{decMult}
\opV'(\sigma)=\bigoplus_{\alpha'} \text{\noindent
\(\mathds{1}\)}_{\alpha'}^{U(d)}\ot \widetilde{\opV}'_{\alpha'}(\sigma),
\ee
where $\alpha'$ runs over two possible kinds of partitions: $(\alpha_{\mathcal{N}},\ldots,-1)$ or $(\alpha_{\mathcal{M}},\ldots,0)$ of length $d$. One can see that we have some abuse of notation here. In this section we want to establish some correspondence between our work and method of calculations of irreps presented in~\cite{Kimura1}. In most cases length of $\alpha_{\mathcal{N}}$ and $\alpha_{\mathcal{M}}$ is smaller than $d$, so in empty places we put zeros.  Suppose that we have some $\alpha_{\mathcal{N}}=(2,1)$ for let us say $d=5$. That notation $(\alpha_{\mathcal{N}},\ldots,-1)$  is nothing else like $(2,1,0,0,-1)$ with totally length equal to $d=5$. This same situation we have of course for partitions $\alpha_{\mathcal{M}}$. Suppose that $\alpha_{\mathcal{M}}=(1,1)$ and $d=2$, then in this case we have $(1,1)$ with totally length equal to two.\\
Form formula~\eqref{decMult} it follows directly that multiplicity of any irrep $\widetilde{\opV}'_{\alpha'}(\sigma)$ is given by dimension of unitary group $U(d)$(see for example~\cite{Bump},~\cite{Sepanski},\cite{Roy1}).
Let us recall here this theorem which tells us how to find $\operatorname{dim}U(d)$ for given irrep $\alpha$:
\begin{theorem}
\label{DDim}
The irreducible representations of $U(d)$ may be indexed by non increasing length-$d$ integer sequences: $\alpha=(\alpha_1,\ldots,\alpha_d)$. If irreducible representation is indexed by $\alpha$, then its dimension is 
\be
\label{dim}
d_{\alpha}=\prod_{1\leq i<j\leq d}\frac{\alpha_i-\alpha_j+j-i}{j-i}.
\ee
\end{theorem}
Thanks to this we obtain decomposition any operator $\opV'(\sigma)$ into irreps together with multiplicities.
\begin{example}
\label{exx}
We present here one example how to apply above discussion to obtain required multiplicities. Let us consider case when $n=4$ and $d=4$. Thanks to equation~\eqref{decMult} we know that allowed partitions $\alpha'$ of length $d$ belong to the set
\be
\label{ir}
\begin{split}
\{\alpha_1',\alpha_2',\alpha_3',\alpha_4',\alpha_5'\}&=\{\underbrace{(3,0,0,-1),(2,1,0,-1),(1,1,1,-1)}_{\text{irreps from $\mathcal{H}_{\mathcal{N}}$}},\underbrace{(2,0,0,0),(1,1,0,0)}_{\text{irreps from $\mathcal{H}_{\mathcal{M}}$}}\},\\
\{\alpha_{\mathcal{N}}^{(1)},\alpha_{\mathcal{N}}^{(2)},\alpha_{\mathcal{N}}^{(3)}\}&=\{(3),(2,1),(1,1,1)\}, \qquad \  \{\alpha_{\mathcal{M}}^{(1)},\alpha_{\mathcal{M}}^{(2)}\}=\{(2),(1,1)\}.
\end{split}
\ee
Using directly Theorem~\ref{DDim} we have multiplicities of irreps from~\eqref{ir}:
\be
\begin{split}
& \mult(\alpha_1')=70,\quad \mult(\alpha_2')=64, \quad \mult(\alpha_3')=10,\\
& \mult(\alpha_4')=10, \quad \mult(\alpha_5')=6.
\end{split}
\ee
One can see that we can easy check correctness of above results. We know that dimension of full Hilbert space in this case is equal to $d^n=256$. Now calculate sum over dimensions of all irrpes and taking into account their multiplicities we have:
\be
70 \cdot 1+64 \cdot 2+10\cdot 1+10\cdot 3+6\cdot 3=256.
\ee 
\end{example}

\begin{example}
We present here one example how to apply above discussion to obtain required multiplicities but for more tricky case than in the Example~\ref{exx}. Namely we illustrate method from this section for the case when some irreps do not occur in decomposition~\eqref{decMult}. Let us consider $n=4$ and $d=2$ which is the same case like in Example~\ref{eexx}. All method are the same like in Example~\ref{exx} but here we have discard these irreps for which number of rows is smaller or equal to $d=2$. In other words multiplicity of discarded irreps is equal to zero. Wee see that these irreps are those labelled by $\alpha_{\mathcal{N}}^{(2)}$ and $\alpha_{\mathcal{N}}^{(3)}$. Thanks to this decomposition~\eqref{ir} takes a form
\be
\label{ir2}
\begin{split}
& \{\alpha_1',\alpha_2',\alpha_3',\alpha_4',\alpha_5'\}=\{\underbrace{(3,-1)}_{\text{irreps from $\mathcal{H}_{\mathcal{N}}$}},\underbrace{(2,0),(1,1)}_{\text{irreps from $\mathcal{H}_{\mathcal{M}}$}}\},\\
& \{\alpha_{\mathcal{N}}^{(1)}\}=\{(3)\}, \qquad \  \{\alpha_{\mathcal{M}}^{(1)},\alpha_{\mathcal{M}}^{(2)}\}=\{(2),(1,1)\}.
\end{split}
\ee
Using directly Theorem~\ref{DDim} we have multiplicities of irreps from~\eqref{ir}:
\be
\begin{split}
\mult(\alpha_1')=5, \quad \mult(\alpha_4')=3, \quad \mult(\alpha_5')=1.
\end{split}
\ee
One can see that we can easy check correctness of above results. We know that dimension of full Hilbert space in this case is equal to $d^n=16$. Now calculate sum over dimensions of all irrpes and taking into account their multiplicities we have:
\be
5 \cdot 1+3 \cdot 3+1\cdot 2=16.
\ee 
Note that we have to use here result from Example~\ref{eexx} which tells us that dimension of irrep labeled by $(1,1)$ is equal to two.
\end{example}

\subsection{Exemplary application}
\label{application}
 In this paper we present solution of our main problem for the case when we have only one term $U^*$. It is enough to investigate when $n-$party $U^{\ot (n-1)}\ot U^*$ invariant states are PPT states respect to transposition on last subsystem. Let us mention here that following example exactly correspond with the simplest case from~\cite{Eggeling1}, so we present here only the crucial steps how to use methods from our work to obtain correct results. As an example let us consider sate which is mixture of three Young projectors for the case when $n=3$:
\be
\label{rhoPPT}
\rho=\widetilde{a}_{\lambda_1}\opP_{\lambda_1}+\widetilde{a}_{\lambda_2}\opP_{\lambda_2}+\widetilde{a}_{\lambda_3}\opP_{\lambda_3},
\ee
where $\lambda_1=(1,1,1)$, $\lambda_2=(2,1)$, $\lambda_3=(3)$ and $\{\widetilde{a}_{\lambda_i}\}_{i=1}^3$  are positive  coefficients. We know that projector $\opP_{\lambda_1}$ projects onto antisymmetric subspace, $\opP_{\lambda_3}$ onto symmetric subspace and finally projector $\opP_{\lambda_2}$ projects onto some nontrivial subspace. 
According to~\cite{Eggeling1} we redefine coefficients in equation~\eqref{rhoPPT} using formulas:
\be
\begin{split}
a_{\lambda_1}&=\frac{1}{6}d(d-1)(d-2)\widetilde{a}_{\lambda_1},\\
a_{\lambda_2}&=\frac{2}{3}d(d^2-1)\widetilde{a}_{\lambda_3},\\
a_{\lambda_3}&=\frac{1}{6}d(d+1)(d+2)\widetilde{a}_{\lambda_2},
\end{split}
\ee
where dimension of local Hilbert space satisfies condition $d>2$, because then $\operatorname{dim}\opP_{\lambda_1}\neq 0$.
Because $\rho$ is density operator we assume that $\rho \geq 0$ and $\tr \rho=1$, so we obtain first set of the conditions on the coefficients $\{a_{\lambda_i}\}_{i=1}^3$. Additionally we require that our density operator $\rho$ have a PPT property respect to transposition over last subsystem. To ensure this we have to have $\rho'\geq 0$, and this gives us second set of conditions on coefficients $\{a_{\lambda_i}\}_{i=1}^3$. We know that $\rho$ is mixture of permutation operator and $\rho'$ is mixture of its partial transpositions, so we can represent every operator $\opV'(\sigma)$ in our operator basis from Definition~\ref{Definitdeff0} ( exactly like in Section~\ref{s:examples} ) and  get quite easily eigenvalues of $\rho'$. Now combining all condition we reduce our problem from 3-dimensional problem to 2-dimensional problem and obtain the region like in Figure~\ref{fig1} of all allowed values of coefficients $a_{\lambda_2}, a_{\lambda_1}$. 
\begin{figure}[ht!]
\includegraphics[width=0.4\textwidth, height=0.3\textheight]{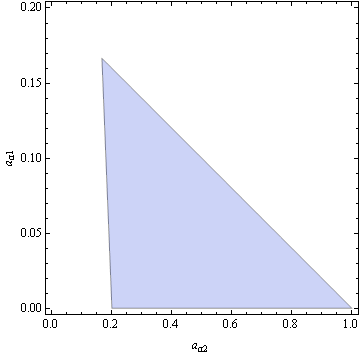}
\caption{\textbf{(Color online) Allowed region of coefficients $a_{\lambda_1}$ and $a_{\lambda_2}$ for which our $U^{\ot 2}\ot U^*$ states have PPT property. Let us notice that this result is dimension independent for any $d>2$. It corresponds exactly with result from~\cite{Eggeling1}.}
\label{fig1}}
\end{figure}

\section{Conclusions}
In this paper we have shown how to construct matrix elements of irreducible representations of partially transposed permutation operators on last subsystem (Lemma~\ref{ElementsIrreps}). Our considerations are  inspired by Schur-Weyl duality and theory of Walled Brauer Algebras. We give here direct method of construction for an arbitrary number of subsystems $n$ and local dimension $d$. Our results come from  very simple observation, that every operator $\opV(\sigma)$ after partial transposition is some mixture of maximally entangled states every time when $\sigma(n)\neq n$. This observation allow us to construct nonorthogonal basis of vectors~\eqref{n10} from which using detailed discussion about Gram matrices (Definition~\ref{Q}, Theorem~\ref{theoremForMatrix}) we are able to built operator basis $\omega_{ij}^{ab}(\alpha)$, orthogonal in Hilbert-Schmidt norm (see Definition~\ref{Definitdeff0} and Lemma~\ref{CompOfOmegas}). 
We can ask here is it possible to generalize this method for larger number of conjugations, or in other words for bigger number of partial transpositions? It seems that the answer is positive, but already two transpositions make our situation much more complicated, so problem is still open.

\section{Acknowledgment}
Authors would like to thank Aram Harrow and Mary Beth Ruskai for valuable discussion. M.S. would like to thank also to Institute for Theoretical Physics, University of Wroc{\l}aw for hospitality where some part of this work was done. M. S. is supported by the International PhD Project "Physics of future quantum-based information technologies": grant MPD/2009-3/4 from Foundation for Polish Science. Authors are also supported by ERC grant QOLAPS (291348). 
Part of this work was done in National Quantum Information Centre of Gda\'nsk.

\section{Appendix}
\subsection{Proofs of the theorems from chapter~\ref{OurResults}}
\label{FullDescription}
This section we start from basic definition for further considerations:
\begin{definition}
\label{defIso}
For given permutation $\pi_k\in S(n-1)$, where $k=1,\ldots,n-1$ which satisfies $\pi(n-1)=k$ we define map $\operatorname{F}_{k}:\mathcal{H}^{\ot (n-2)}\rightarrow \mathcal{H}^{\ot n}$ as follows
\be
\operatorname{F}_k|\phi\>=\sqrt{d}\opV(\pi_k)|\phi\>|\Phi_+\>_{n-1,n},
\ee
where $|\Phi_+\>=\frac{1}{\sqrt{d}}\sum_{l=1}^d|ll\>_{n-1,n}$ is maximally entangled state.
\end{definition}
Note that for every fixed natural number $k$ we can define many maps $\operatorname{F}_k$ because we have only weak constraint that $\pi(n-1)=k$ and it is clear that we can find many permutation from $S(n-1)$ satisfy this property. Reader notices that mapping defined in equation~\eqref{eq:cos} is simply composition of maps $\operatorname{F}_k$ according to formula $\mathcal{F}_{ab}^t(\cdot) \equiv \operatorname{F}_b (\cdot) \operatorname{F}_a^{-1}$. Now if we consider basis of irreps  $|\phi(\alpha,r)\>$, where $r$ is multiplicity and $\alpha$ is partition we can write:
\be
\label{vveeccss}
|\psi_i^k(\alpha,r)\>=\operatorname{F}_k|\phi_i(\alpha,r)\>.
\ee
Moreover we have following

\begin{definition}
Let  $\mathcal{H}_M^{\alpha, r}$ be linear space defined as follows
\be
 \mathcal{H}_M^{\alpha, r}=\operatorname{span}_{\mathbb{C}}\left\{|\psi_i^k(\alpha,r)\> \ | \ |\psi_i^k(\alpha,r)\> \in S_M^{\alpha, r}\right\},
\ee
where
\be
S_M^{\alpha, r}=\left\{|\psi_{i}^k(\alpha,r)\>\in \mathcal{H}^{\ot n} \ | \ 1\leq k \leq n-1, |\phi_i(\alpha,r)\>\in\mathcal{H}^{\ot n-2}\right\}.
\ee
\end{definition}
Above Hilbert space is invariant under action of group $S(n-1)$ which is crucial property in our construction. 

Next important property which we would like to now is how to calculate scalar product between vectors 
$|\psi_i^k(\alpha,r)\>$ from the same invariant subspace $\mathcal{H}_M^{\alpha, r}$ i.e. labeled by the same index $\alpha$.  We can find answer in next lemma:

\begin{replemma}{Lemmailsk}
\lemmailsk
\end{replemma}

\begin{proof}
We prove our statement by direct calculations using Definition~\ref{defIso}:
\be
\begin{split}
&\<\psi_i^a(\alpha,r)|\psi_j^b(\alpha,r)\>=d\cdot \<\Phi_+|\<\phi_i(\alpha,r)|\Pi_{a}^{-1}\Pi_b|\phi_j(\alpha,r)\>\Phi_+\>=
d\tr\left(|\Phi_+\>\<\Phi_+|\ot |\phi_i(\alpha,r)\>\<\phi_j(\alpha,r)|\Pi_a^{-1}\Pi_b\right)=\\
&=\tr\left(\Pi_a^{-1}\Pi_b\opV(nn-1) \text{\noindent
\(\mathds{1}\)}_{nn-1} \ot  |\phi_i(\alpha,r)\>\<\phi_j(\alpha,r)|\right)=\<\phi_i(\alpha,r)|\opX(\chi_{ab})|\phi_j(\alpha,r)\>
\end{split}
\ee
\end{proof}

It is worth to say a few words about properties of  operator $\opX(\chi_{ab})$ for different relations between indices $a$ and $b$. 

\begin{replemma}{Lemmaprop}
\lemmaprop
\end{replemma}

\begin{proof}
Let us define permutation $\eta_{ab}:= (nn-1)\circ \pi_a^{-1}\circ \pi_b$. Now our goal is to find partial trace over two last subsystems of $\opX(\eta_{ab})$. Let us assume here that $a\neq b$ then we deal with permutation operator $\opV(\eta_{ab})$. To calculate partial trace over some permutation operator we have to decompose permutation into disjoint cycles and then to obtain correct result we have to simply discard subsystems on which trace is calculated. To do this we have to decompose permutation $\eta_{ab}$ into disjoint cycles  and find cycle which contains numbers $n-1$ and $n$. First let us find  $x\in \mathbb{N}$ such that $\eta_{ab}[x]=n$, so $x=\eta_{ab}^{-1}[n]$, then
\be
x=\pi^{-1}_b\circ \pi_a\circ(nn-1)[n]=\pi^{-1}_b\circ \pi_a[n-1]=\pi^{-1}_b[a],  \ \text{and} \  x=n-1, \quad \text{when} \ a=b.
\ee
  After decomposition permutation $\eta_{ab}$ into disjoint cycles we have
\begin{equation}
\eta_{ab} =\eta'_{ab}\circ (xnn-1)\qquad  a \neq b,
\end{equation}
where  $\eta'_{ab}$  is some combination of disjoint cycles which do not contain $n$ and $n-1$. Now let us lift permutations into permutation operators, then partial trace over $\opV(\eta_{ab})$ is equal to $\opV(\chi_{ab})=\opV(\eta'_{ab})$. Now, since $\opV(\eta_{ab})=\opV(\eta_{ab}')\opV(xnn-1)$ we get $\opV(\eta_{ab}')=\opV(\eta_{ab})\opV(n-1nx)$. Finally partial trace from operator $\opV(\eta_{ab})$ for the case $a=b$ is equal to:
\begin{equation}
\opV(\chi_{ab})=\tr_{n,n-1}\opV(\eta_{ab})= 
\opV(\eta_{ab})\opV(n-1n\pi^{-1}_b[a]).
\end{equation}
 One can see that for $a=b$ situation is much more easier, because partial trace is simply equal to $d \cdot \text{\noindent\(\mathds{1}\)}_{1\ldots n-2}$, so we obtain required formula for both cases $a=b$ and $a\neq b$.
At the end we show for the case when $a\neq b$  any $\chi_{ab}$ is permutation at most from $S(n-2)$ :
\be
\begin{split}
&\chi_{ab}[n]=(nn-1)\circ \pi_a^{-1}\circ \pi_b\left[\pi_b^{-1}[a]\right]=(nn-1)\circ \pi_a^{-1}[a]=(nn-1)[n-1]=n\\
&\chi_{ab}[n-1]=(nn-1)\circ \pi_a^{-1}\circ \pi_b[n]=(nn-1)[n]=n-1.
\end{split}
\ee
\end{proof}

\begin{repdefinition}{Definitv0}
\definitv0
\end{repdefinition}

\begin{replemma}{Lemmavij}
\lemmavij
\end{replemma}

\begin{proof}
Using Definition~\ref{Definitv0} and Lemma~\ref{Lemmailsk} we can prove statement of this lemma  by direct calculations:
\be
\begin{split}
v_{ij}^{ab}(\alpha)v_{kl}^{cd}(\beta)&=\sum_r\sum_{r'}|\psi_i^a(\alpha,r')\>\<\psi_j^b(\alpha,r')|\psi_k^c(\beta,r)\>\<\psi_l^d(\beta,r)|=\\
&=d^{\delta_{bc}}\sum_r |\psi_i^a(\beta,r)\>\<\psi_l^d(\beta,r)|\<\phi_j(\beta,r)|\chi_{bc}|\phi_k(\beta,r)\>=\\
&=d^{\delta_{bc}}\varphi_{jk}^{\beta}(\chi_{bc})v_{il}^{ad}(\beta).
\end{split}
\ee
\end{proof}
Let us now lift embedding $S(n-2)$ into $S(n)$.
\begin{lemma}
\label{decomp1}
Partially transposed permutation operator $V'(\sigma_{ab})$ can be decompose as follows:
\be
\opV'(\sigma_{ab})=\Pi_b\opV(f_{ab}^{-1}(\sigma_{ab})) \opV'(nn-1)\Pi_a^{-1},
\ee
where $f_{ab}^{-1}(\sigma_{ab})=\pi_b^{-1}\circ \sigma_{ab}\circ \pi_a\circ (nn-1)$ and $f_{ab}^{-1}(\sigma_{ab})\in S(n-2)$.
\end{lemma}

\begin{proof}
In proof of this lemma we use fact that that permutations from $S(n-1)$ are invariant under action
of partial transposition $'$ over last subsystem. Thanks to this we can write:
\be
\begin{split}
&\Pi_b \opV(f_{ab}^{-1}(\sigma_{ab})) \opV'(nn-1) \Pi_a^{-1}=\opV'(\pi_b \circ f_{ab}^{-1}(\sigma_{ab})\circ (nn-1)\circ \pi_a^{-1})=\\
&=\opV'(\pi_b \circ (\pi_b^{-1}\circ \sigma_{ab}\circ \pi_a\circ (nn-1)) \circ (nn-1) \circ \pi_a^{-1})=\opV'(\sigma_{ab}).
\end{split}
\ee
We can see that $\sigma=\pi_b^{-1}\circ \sigma_{ab}\circ \pi_a\circ (nn-1)$, then using definition of function $f_{ab}(\sigma)$ from equation~\eqref{eq:maps} we have
\be
f_{ab}(\sigma)=f_{ab}(\pi_b^{-1}\circ \sigma_{ab}\circ \pi_a\circ (nn-1))=\sigma_{ab},
\ee
so indeed $\pi_b^{-1}\circ \sigma_{ab}\circ \pi_a\circ (nn-1)=f_{ab}^{-1}(\sigma_{ab})$.
At the end one can see that  $f_{ab}^{-1}(\sigma_{ab})\in S(n-2)$, because $f_{ab}^{-1}(\sigma_{ab})[n]=n$ and $f_{ab}^{-1}(\sigma_{ab})[n-1]=n-1$. This finishes the proof.
\end{proof}
Now we are ready to formulate main result of this paper, namely we have following:

\begin{reptheorem}{ccomp1}
\theoremForVe
\end{reptheorem}

\begin{proof}
To prove this theorem let us recall formulas for algebra $\mathbb{C}[S(n-2)]$ from equation~\eqref{sim2}:
\be
\opE^{\alpha}_{ij}=\frac{d_{\alpha}}{(n-2)!}\sum_{\sigma \in S(n-2)}\varphi_{ij}^{\alpha}(\sigma)\opV(\sigma)\in \mathbb{C}[S(n-2)], \quad  \opV(\sigma)=\bigoplus_{\alpha}\sum_{ij}\varphi^{\alpha}_{ij}(\sigma)\opE_{ij}^{\alpha}\in \mathbb{C}[S(n-2)].
\ee
Applying mapping $\mathcal{F}_{ab}^t$ from Definition~\ref{Fab} to both sides of the above equalities we obtain respectively:
\be
v_{ij}^{ab}(\alpha)=\frac{d_{\alpha}}{(n-2)!}\sum_{\sigma\in S(n-2)}\varphi_{ij}^{\alpha}(\sigma)\opV'(f_{ba}(\sigma))\in M,\quad \opV'(\sigma_{ab})=\bigoplus_{\alpha}\sum_{ij} \varphi_{ij}^{\alpha}(f_{ab}^{-1}(\sigma_{ab}))v_{ij}^{ba}(\alpha)\in M
\ee
This finishes the proof.
\end{proof}

\begin{replemma}{actionOnvij}
\ActionOnvij
\end{replemma}

\begin{proof}
First let us consider elements from algebra $\mathcal{M}$. Using statement from Theorem~\ref{ccomp1} and composition law of operators $v_{ij}^{ab}(\alpha)$ from
Lemma~\ref{Lemmavij} we obtain:
\be
 \begin{split}
 \opV'(\sigma_{ab})v_{kl}^{cd}(\beta)&=\bigoplus_{\alpha}\sum_{ij} \varphi^{\alpha}_{ij}(f_{ab}^{-1}(\sigma_{ab}))v_{ij}^{ba}(\alpha)v_{kl}^{cd}(\beta)=d^{\delta_{ac}}\sum_{ij}\varphi_{ij}^{\alpha}(f_{ab}^{-1}(\sigma_{ab}))\varphi_{jk}^{\alpha}(\chi_{ac})v_{il}^{bd}(\alpha)=\\
 &=d^{\delta_{ac}}\sum_i\varphi_{ik}^{\alpha}\left(f_{ab}^{-1}(\sigma_{ab})\circ \chi_{ac}\right)v_{il}^{bd}(\alpha).
 \end{split}
\ee
 Now we prove thesis for elements from $\mathcal{N}$. Using Definition~\ref{Fab} of $\mathcal{F}_{ab}^t$ we have
 \be
 \begin{split}
 \opV(\sigma)v_{ij}^{ab}(\alpha)&=\opV(\sigma)\mathcal{F}_{ba}^t(\opE_{ij}^{\alpha})=\opV(\sigma)\Pi_a\opE_{ij}^{\alpha}\opV'(nn-1)\Pi_b^{-1}=\Pi_{\sigma[a]}\opV(\xi)\opE_{ij}^{\alpha}\opV'(nn-1)\Pi_b^{-1}=\\
 &=\Pi_{\sigma[a]}\sum_{k}\varphi_{ki}^{\alpha}(\xi)\opE_{kj}^{\alpha}\opV'(nn-1)\Pi_b^{-1}=\sum_k \varphi^{\alpha}_{ki}(\xi)\mathcal{F}_{b\sigma[a]}^t (\opE_{kj}^{\alpha})=\sum_k\varphi_{ki}^{\alpha}(\xi)v_{kj}^{\sigma[a]b}.
 \end{split}
 \ee
 Because of condition $\sigma \circ \pi_a = \pi_{\sigma[a]}\circ \xi$ we get $\xi=\pi_{\sigma[a]}^{-1}\circ \sigma \circ \pi_a \equiv f_a(\sigma)\in S(n-2)$.
This finishes the proof.
\end{proof}
Our next step is to construct operator basis which is orthonormal in Hilbert-Schmidt. First we prove the following theorem which connects invertibility of matrix $Q_{ij}^{ab}(\alpha)$ from Definition~\ref{Q} with plays crucial role in our further construction.

\begin{reptheorem}{theoremForMatrix}
\theoremForMatrix
\end{reptheorem}

\begin{proof}
From the assumption we have $\varphi ^{\alpha }(ij)=\varphi ^{\alpha
}(ij)^{-1}=\varphi ^{\alpha }(ij)^{+}$ so the matrices $\varphi ^{\alpha
}(ij)$ are unitary and hermitean. From this we get 
\[
||\varphi ^{\alpha }(ij)||\equiv \sqrt{\rho (\varphi ^{\alpha }(ij)\varphi
^{\alpha }(ij)^{+})}=1,
\]%
where $\rho (A)$ is the spectral radius of the matrix $A.$ Now we have to
show that 
\[
\forall x^{+}=(x_{1}^{+},x_{2}^{+},...,x_{n-1}^{+})\in 
\mathbb{C}
^{(n-1)w^{\alpha }}\qquad x^{+}Q(\alpha )x>0. 
\]%
Using the explicite block structure of the matrix $Q(\alpha )$ we get 
\[
x^{+}Q(\alpha
)x=\sum_{i=1}^{n-1}dx_{i}^{+}x_{i}+\sum_{i=1}^{n-2}x_{i}^{+}x_{n-1}+%
\sum_{i=1}^{n-2}x_{n-1}^{+}x_{i}+\sum_{i\neq j=1}^{n-2}x_{i}^{+}\varphi
^{\alpha }(ij)x_{j}. 
\]%
It may be written%
\[
x^{+}Q(\alpha
)x=d\sum_{i=1}^{n-1}||x_{i}||^{2}+%
\sum_{i=1}^{n-2}[(x_{i},x_{n-1})+(x_{n-1},x_{i})]+\sum_{i<j}^{n-2}[(x_{i},%
\varphi ^{\alpha }(ij)x_{j})+(x_{j},\varphi ^{\alpha }(ij)x_{i})]. 
\]%
where $(x,y)\equiv x^{+}y$ for $x,y\in 
\mathbb{C}
^{w^{\alpha }}.$ Using the Schwartz inequality and the property $||\varphi
^{\alpha }(ij)||=1$ we get 
\[
|\sum_{i=1}^{n-2}[(x_{i},x_{n-1})+(x_{n-1},x_{i})]|+|%
\sum_{i<j}^{n-2}[(x_{i},\varphi ^{\alpha }(ij)x_{j})+(x_{j},\varphi ^{\alpha
}(ij)x_{i})]|\leq 2\sum_{i<j}^{n-1}||x_{i}||||x_{j}|| 
\]%
and from the elementary inequality%
\[
2ab\leq a^{2}+b^{2} 
\]%
one obtain 
\[
2\sum_{i<j=1}^{n-1}||x_{i}||||x_{j}||\leq
(n-2)\sum_{i=1}^{n-1}||x_{i}||^{2}. 
\]%
Now the application of the inequality 
\[
\forall a\in 
\mathbb{R}
\quad a\geq -|a| 
\]%
gives the inequality 
\[
x^{+}Q(\alpha )x\geq (d-n+2)\sum_{i=1}^{n-1}||x_{i}||^{2}. 
\]
\end{proof}
Using above considerations let us formulate following

\begin{repdefinition}{Definitdeff0}
\definitdeff0
\end{repdefinition}

Now we write all crucial properties of our new set of operators $\{\omega_{ij}^{ab}(\alpha)\}$, such as rule of composition, action on elements from algebras $\mathcal{M}$ and $\mathbb{C}[S(n-1)]$, rewrite Theorem~\ref{ccomp1} and finally we calculate matrix elements of irreps.
\begin{replemma}{CompOfOmegas}
\compOfOmegas
\end{replemma}

\begin{proof}
Using exactly Lemma~\ref{aux2} from Appendix~\ref{opBases3} we get that the transformation rules are the same as for operators $v_{ij}^{ab}(\alpha)$ in Lemma~\ref{Lemmavij}.
\end{proof}

\begin{replemma}{actionOnF}
\ActionOnF
\end{replemma}

\begin{proof}
Our proof we start for elements from $\mathcal{M}$. We know that operators $\omega_{ij}^{ab}(\alpha)$ are given by some linear transformation of operators $v_{ij}^{ab}(\alpha)$ (see Definition~\ref{Definitdeff0}), so using exactly Lemma~\ref{aux} from Appendix~\ref{opBases3} and multiindeces we obtain required rule of transformation, which is the same like for operators $v_{ij}^{ab}(\alpha)$.
To prove statement of lemma for elements from $\mathbb{C}[S(n-1)]$ we use conclusions from  Lemma~\ref{actionOnvij} and Definition~\ref{Definitdeff0}. Namely we have:
\be
\begin{split}
\opV(\sigma)\omega_{ij}^{ab}(\alpha)&=\sum_{mn}D_{jn}^{bm}(\alpha)\opV(\sigma)v_{in}^{am}(\alpha)=\sum_{mn}D_{jn}^{bm}(\alpha)\sum_k \varphi^{\alpha}_{ki}(f_a(\sigma))v_{kn}^{\sigma[a]m}=
\sum_k \varphi^{\alpha}_{ki}(f_a(\sigma)) \omega_{kj}^{\sigma[a]b}.
\end{split}
\ee
This finishes the proof. 
\end{proof}
To rewrite Theorem~\ref{ccomp1}
in terms of new operators $\omega_{ij}^{ab}(\alpha)$ we need to know how to represent operators $v_{ij}^{ab}(\alpha)$ like a function of $\omega_{ij}^{ab}(\alpha)$. This is given in the following
\begin{fact}
Operators $v_{ij}^{ab}(\alpha)$ can be written in terms of $\omega_{ij}^{ab}(\alpha)$ as follows:
\be
\label{invers}
v_{ij}^{ab}(\alpha)=\sum_{kc}\varphi_{jk}^{\alpha}(\chi_{bc})\omega_{ik}^{ac}(\alpha).
\ee
\end{fact}
\begin{proof}
 Let us put equation~\eqref{deff0} from Definition~\ref{Definitdeff0} directly to equation~\eqref{invers}:
\be
\sum_{kc}\varphi_{jk}^{\alpha}(\chi_{bc})\omega_{ik}^{ac}(\alpha)=\sum_{kc}\sum_{mn} \varphi_{jk}^{\alpha}(\chi_{bc})D_{kn}^{cm}(\alpha)v_{in}^{am}(\alpha)=\sum_{mn}\delta^{bm}\delta_{jn}v_{in}^{am}(\alpha)=v_{ij}^{ab}(\alpha).
\ee
We have identity, so proof is finished.
\end{proof}

\begin{reptheorem}{thForOmegas}
\theoremForOmegas
\end{reptheorem}

\begin{proof}
To prove first part of theorem we put equation~\eqref{eq:Comp1} into Definition~\ref{Definitdeff0} of operators $\omega_{ij}^{ab}(\alpha)$:
\be
\begin{split}
\omega_{ij}^{ab}(\alpha)&=\sum_{kc}D_{jk}^{bc}(\alpha)v_{ik}^{ac}(\alpha)=\frac{d_{\alpha}}{(n-2)!}\sum_{kc}\sum_{\sigma\in S(n-2)} D_{jk}^{bc}(\alpha)\varphi^{\alpha}_{ik}(\sigma)\opV'(f_{ca}(\sigma))=\\
&=\frac{d_{\alpha}}{(n-2)!}\sum_{kc}\sum_{\sigma\in S(n-2)}\varphi_{ik}^{\alpha}(\sigma)\varphi_{kj}^{\alpha}(\chi_{bc}^{-1})\opV'(f_{ca}(\sigma))=\frac{d_{\alpha}}{(n-2)!}\sum_{c}\sum_{\sigma\in S(n-2)}\varphi_{ij}^{\alpha}(\sigma \circ \chi_{bc}^{-1})\opV'(f_{ca}(\sigma)).
\end{split}
\ee
To prove second part we put directly equation~\eqref{invers} to equation~\eqref{eq:decomp2} from Theorem~\ref{ccomp1}
\be
\opV'(\sigma_{ab})=\bigoplus_{\alpha}\sum_{ij}\varphi_{ij}^{\alpha}(f_{ab}^{-1}(\sigma_{ab}))\sum_{kc}\varphi_{jk}^{\alpha}(\chi_{ac})\omega_{ik}^{bc}(\alpha)=\bigoplus_{\alpha}\sum_{ikc}\varphi_{ik}^{\alpha}(f_{ab}^{-1}(\sigma_{ab})\circ \chi_{ac})\omega_{ik}^{bc}(\alpha).
\ee
Making the same argumentation like in proof of Theorem~\ref{ccomp1} we see that $\omega_{ij}^{ab}(\alpha)\in \mathcal{M}$. This finishes the proof.
\end{proof}

\begin{replemma}{ElementsIrreps}
\irreps
\end{replemma}

\begin{proof}
In this proof we use explanations from Appendix~\ref{A3} for special case, i.e. for $\opV'(\sigma_{ab})$ and $\opV(\sigma)$. Thus, for this case using Lemma~\ref{actionOnF} we have
\be
\begin{split}
[\widetilde{\opV}'_{\alpha}(\sigma_{ab})]_{dl,ck}&=\frac{1}{m_{\alpha}}\tr[\opV'(\sigma_{ab})\omega_{kl}^{cd}(\alpha)]=\frac{d^{\delta_{ac}}}{m_{\alpha}}\sum_i \varphi_{ik}^{\alpha}\left( f^{-1}_{ab}(\sigma_{ab})\circ \chi_{ac}\right)\tr[\omega_{il}^{bd}(\alpha )]=\\
&=d^{\delta_{ac}}\sum_i \varphi_{ik}^{\alpha}\left( f^{-1}_{ab}(\sigma_{ab})\circ \chi_{ac}\right)\delta_{bd} \delta_{il}=d^{\delta_{ac}}\varphi_{lk}^{\alpha}(f^{-1}_{ab}(\sigma_{ab})\circ \chi_{ac})\delta_{bd}.\\
[\widetilde{\opV}_{\alpha}(\sigma)]_{dl,ck}&=\frac{1}{m_{\alpha}}\tr[\opV'(\sigma)\omega_{kl}^{cd}(\alpha)]=\frac{1}{m_{\alpha}}\sum_i\varphi_{ik}^{\alpha}(f_a(\sigma ))\tr[\omega_{il}^{\sigma[c]d}]=\\
&=\sum_i \varphi_{ik}^{\alpha}(f_a(\sigma ))\delta_{\sigma[c]d}\delta_{il}=\varphi_{lk}^{\alpha}(f_a(\sigma ))\delta_{\sigma[c]d}.
\end{split}
\ee
In above calculations we used condition of biorthogonality which gives us that $\tr[\omega_{ij}^{ab}(\alpha)]=\frac{1}{m_{\alpha}}\delta_{ab}\delta_{ij}$, where number $m_{\alpha}$ is multiplicity of irrep $\alpha$.
\end{proof}

\subsection{Operator bases}
\label{opBases}
In this section we present short discussion how to construct biorthogonal basis from nonorthogonal one. This exactly illustrate our method from the Subsections~\ref{nonort} and~\ref{bioort} of this paper, but here we do it for an arbitrary basis, so this considerations are valid for general case. Here operator $\w_{ij}$ will be prototypes of operators $v_{ij}^{ab}(\alpha)$ from Definition~\ref{Definitv0} and operators $\ww_{ij}$ will play role of operators $\omega_{ij}^{ab}(\alpha)$ from Definition~\ref{Definitdeff0}.\\
Given a Hilbert space $\hcal$  consider a set  $S_M=\{\psi_i\}_{i=1}^k$ of vectors $\psi_i\in \hcal$ (possibly nonnormalised, and nonorthonormal). Let us construct the following operators that are poor versions of
operators $|i\>\<j|$, for an orthonormal set $\{|i\>\}$ of vectors:
\be
\label{eq:o1}
\w_{ij}=|\psi_i\>\<\psi_j|.
\ee
The operators $\w_{ij}$ satisfy following composition rule:
\be
\label{eq:compw}
\w_{ij}\w_{kl}=|\psi_i\>\<\psi_j|\psi_k\>\<\psi_l|= G_{jk}\w_{il},
\ee
where $G_{jk}$ is matrix element of Gram matrix $G$ for vectors $\{|\psi_i\>\}$:
\be
G_{jk}=\<\psi_j|\psi_k\>.
\ee 
Since the set of vectors $\psi_i$ is linearly independent, the Gram matrix is nonsingular,
hence there exists inverse matrix, which we will denote by $D$.
Using this matrix we can now define new operators, which have property of composition,
the same as the mentioned operators $|i\>\<j|$. 
\be
\label{eq:e}
\ww_{ij}=\sum_k D_{jk} \w_{ik}=|\psi_i\>\<\phi_j|,
\ee
where $D_{jk}$ is matrix element of matrix $D = G^{-1}$ and $|\phi_j\>=\sum_k D_{kj}|\psi_k\>$.
One easily finds, that $\ww_{ij}$ satisfy the required composition rule:
\be
\ww_{ij} \ww_{kl} = \delta_{jk} \ww_{il}
\ee
Indeed, using composition rule for operators $\w_{ij}$ from equation~\eqref{eq:compw} 
we get 
\be
\ww_{ij}\ww_{kl}=\sum_{nm} D_{jn} D_{lm} \w_{in}\w_{km}=\sum_{n}D_{jn}G_{nk} \sum_m D_{lm}\w_{im}=\delta_{jk}\ww_{il}.
\ee
We have also the following inverse relations
\be
\label{eq:inv1}
\w_{ij}=\sum_k G_{jk}\ww_{ik}.
\ee
Indeed if we put equation~\eqref{eq:e} into right hand side  of formula~\eqref{eq:inv1} we obtain
\be
\sum_k G_{jk} \sum_l D_{kl}\w_{il}=\sum_l \delta_{jl}\w_{il}=\w_{ij}.
\ee

\subsection{From $\w_{ij}$ to $\ww_{ij}$}
\label{opBases2}
In this section we illustrate how to rewrite formulas for operators which are linear combination of operators $\w_{ij}$ in terms of operators $\ww_{ij}$. It is worth to say here that in our work we mostly consider operators of the form $\opE_{ij}=\text{\noindent\(\mathds{1}\)}\ot |i\>\<j|$, but further formulas are still valid for this case. 
In this section $g$ is an arbitrary label and the operators  $\opX(g)$ are prototypes of permutation operators $\opV'(\sigma)$ and then $g$ will be from $S(n)$.\\
Consider now set of operators $\opX(g)$ acting on our Hilbert space $\hcal$,
and suppose that we are given the relations between this set and operators $\w_{ij}$:
\be
\label{eq:f1}
\opX(g):=\sum_{ij}a_{ij}(g)\w_{ij}
\ee
and the inverse $\opX(g)$:
\be
\label{eq:f2}
\w_{ij}=\sum_{g} b_{ij}(g)\opX(g),
\ee
Now our goal is rewrite equation~\eqref{eq:f1} and~\eqref{eq:f2} in terms of orthogonal operators $\ww_{ij}$. 

Now we are ready to prove following
\begin{fact}
For operators $X(g)$ satisfying formulas \eqref{eq:f1} and  \eqref{eq:f2}
we have the following relations:
\be
\opX(g)=\sum_{ik}\left(A(g)G\right)_{ik}\ww_{ik},
\ee
and the inverse relations
\be
\ww_{ij}=\sum_{g} \left(B(g)\bar{G}^{-1}\right)_{ij}\opX(g).
\ee
\end{fact}

\begin{proof}
Putting equation~\eqref{eq:inv1} into formula~\eqref{eq:f1} we get
\be
\opX(g)=\sum_{ij}a_{ij}(g)\w_{ij}=\sum_{ij}a_{ij}(g)\sum_k G_{jk}\ww_{ik}=\sum_{ik}\left(\sum_j a_{ij}(g)G_{jk}\right)\ww_{ik}=\sum_{ik}\left(A(g)G\right)_{ik}\ww_{ik}.
\ee
To prove inverse relation we put equation~\eqref{eq:o1} into equation~\eqref{eq:e}
\be
\ww_{ij}=\sum_k D_{jk}\w_{ik}=\sum_k D_{jk}\sum_g b_{ik}(b) \opX(g)=\sum_g \left(B(g)\bar{G}^{-1}\right)_{ij}\opX(g)
\ee
\end{proof}

\subsection{Auxiliary Lemmas}
\label{opBases3}
We present here some set of auxiliary lemmas which are useful to prove some statements form Section~\ref{FullDescription}. In the lemmas below $X, Y_{kl}, Z_{kl}$ (for any fixed $k$ and $l$) 
are square matrices.
\label{A1}
\begin{lemma}
\label{aux}
Suppose that
\be
\operatorname{X}\operatorname{Y}_{kl}=\sum_m a_{km} \operatorname{Y}_{ml},
\ee
then if we define $\operatorname{Z}_{kl}=\sum_n b_{ln} \operatorname{Y}_{kn}$ we obtain:
\be
\operatorname{X}\operatorname{Z}_{kl}=\sum_j a_{kj} \operatorname{Z}_{jl},
\ee
\end{lemma}

\begin{proof}
By direct calculations:
\be
\begin{split}
\operatorname{X}\operatorname{Z}_{kl}&=\sum_n b_{ln} \operatorname{X}\operatorname{Y}_{kn}=\sum_n b_{ln} \sum_m a_{km} \operatorname{Y}_{mn}=
\sum_m a_{km} \left(\sum_n b_{ln} \operatorname{Y}_{mn} \right)=\sum_m a_{km} \operatorname{Z}_{ml}.
\end{split}
\ee
\end{proof}

\begin{lemma}
\label{aux2}
Suppose that 
\be
\opY_{ij}\opY_{kl}=a_{jk}\opY_{il}
\ee
then if we define $\opZ_{kl}=\sum_nb_{ln}\opY_{kn}$, where $\sum_k a_{ik}b_{kj}=\delta_{ij}$ we obtain:
\be
\opZ_{ij}\opZ_{kl}=\delta_{jk}\opZ_{il}.
\ee
\end{lemma}

\begin{proof}
By direct calculations:
\be
\begin{split}
\opZ_{ij}\opZ_{kl}=\sum_n\sum_m b_{jn}\opY_{in} b_{lm}\opY_{km}=\sum_{n} b_{jn}a_{nk}\sum_m b_{lm}\opY_{im}=\delta_{jk}\opZ_{il}.
\end{split}
\ee
\end{proof}

\subsection{Matrix elements in biorthogonal basis}
\label{A3}
In this section we give explicit formulas for matrix elements of an arbitrary operator $X\in \mathcal{A}$ in biorthogonal basis. It is clear that such elements, let us say $X_{ij}$ are given by equation
\be
\label{mel}
X|\psi_i\>=\sum_j X_{ij}|\phi_j\>.
\ee
Since vectors $|\psi_i\>$ and $|\phi_j\>$ form an biorthogonal basis we get that $X_{ij}=\<\psi_i|X|\phi_j\>$.
Here we restrict to basis in given irrep $\alpha$, but of course results are valid in arbitrary biorthogonal system. We know that we can rewrite operators $X$ and $\omega_{ij}^{ab}(\alpha)$ like identity operator on multiplicity space and the "rest" on representation space for fixed irrep $\alpha$:
\be
X=\bigoplus_{\alpha}\text{\noindent
\(\mathds{1}\)}_{\alpha}\ot \widetilde{X}_{\alpha},\qquad \omega_{ij}^{ab}(\alpha)=\text{\noindent
\(\mathds{1}\)}_{\alpha}\ot |\phi^{\alpha}_{b,j}\>\<\psi^{\alpha}_{a,i}|,
\ee
where vectors $\{|\phi^{\alpha}_{b,j}\>, |\psi^{\alpha}_{a,i}\>\}$ form an biorthonormal basis in representation space. Hence we can write
\be
\tr[X\omega_{ij}^{ab}(\alpha)]=\tr[(\text{\noindent
\(\mathds{1}\)}_{\alpha}\ot \widetilde{X}_{\alpha})( \text{\noindent
\(\mathds{1}\)}_{\alpha}\ot |\phi^{\alpha}_{b,j}\>\<\psi^{\alpha}_{a,i}|)]=m_{\alpha}\<\psi^{\alpha}_{a,i}|\widetilde{X}_{\alpha}\phi^{\alpha}_{b,j}\>,
\ee
where $m_{\alpha}$ is multiplicity of irrep $\alpha$. Finally matrix elements are given by the formula
\be
X_{ai,bj}^{\alpha}=\<\psi^{\alpha}_{a,i}|\widetilde{X}_{\alpha}|\phi^{\alpha}_{b,j}\>=\frac{1}{m_{\alpha}}\tr[X\omega_{ij}^{ab}(\alpha)].
\ee
\subsection{Division into class $S_{ab}$ for $n=3,4,5$}
\label{class}
\begin{example}
\label{ex}
This example illustrates Notation~\ref{not1}. Let us consider two permutations. First one is the element of $S(3)$, second one is element of group $S(4)$.
\bigskip 
\[
\left( 
\begin{array}{ccc}
1 & 2 & 3 \\ 
2 & 3 & 1%
\end{array}%
\right) =\left( 
\begin{array}{ccc}
1 & 2 & 3 \\ 
2 & 3 & 1%
\end{array}%
\right) _{(2,1)},\qquad \left( 
\begin{array}{cccc}
1 & 2 & 3 & 4 \\ 
2 & 1 & 4 & 3%
\end{array}%
\right) =\left( 
\begin{array}{cccc}
1 & 2 & 3 & 4 \\ 
2 & 1 & 4 & 3%
\end{array}%
\right) _{(3,3)}. 
\]
Looking at these examples we see that in the first case we have $a=2, b=1$ and for the second permutation $a=3,b=3$.
\end{example}
To illustrate division into class $S_{ab}$ let us recall here equation~\eqref{division}:
\be
S(n)=\bigcup_{a,b=1}^{n-1}S_{ab}\cup S(n-1).
\ee
\subparagraph{Case $n=3$} In this case $1\leq a,b \leq 2$, so we have
\be
S_{11}=\{(13)\}, \ S_{12}=\{(123)\}, \ S_{21}=\{(132)\}, \ S_{22}=\{(23)\} \ \text{and} \ S(2)=\{ e,(12)\}.
\ee
\subparagraph{Case $n=4$} In this case $1\leq a,b \leq 3$, so we have
\be
\begin{split}
S_{11}&=\{(14),(14)(23)\}, \ S_{12}=\{(124),(1324)\}, \ S_{13}=\{(1234),(134)\},\\
S_{21}&=\{(142),(1423)\}, \ S_{22}=\{(24),(13)(24)\}, \ S_{23}=\{(234),(1342)\},\\
S_{31}&=\{(1432),(143)\}, \ S_{32}=\{(243),(1243)\}, \ S_{33}=\{(34),(12)(34)\}.
\end{split}
\ee
We  see that all classes possess two elements.
\subparagraph{Case $n=5$} In this case $1\leq a,b \leq 4$, so we have
\be
\begin{split}
S_{11}&=\{(15),(15)(34),(15)(23),(15)(234),(15)(243),(15)(24)\},\\
S_{12}&=\{(125),(125)(34),(1325),(13425),(14325),(1425)\},\\
S_{13}&=\{(12345),(12435),(135),(135)(24),(1435),(14235)\},\\
S_{14}&=\{(12345),(1245),(1345),(13245),(145),(145)(23)\},\\
S_{21}&=\{(152),(152)(34),(1523),(15234),(15243),(15241)\},\\
S_{22}&=\{(25),(25)(34),(13)(25),(134)(25),(143)(25),(14)(25)\},\\
S_{23}&=\{(235),(2435),(1352),(13524),(14352),(14)(235)\},\\
S_{24}&=\{(2345),(245),(13452),(13)(245),(1452),(14523)\},\\
S_{31}&=\{(1532),(15342),(153),(1534),(153)(24),(15324)\},\\
S_{32}&=\{(253),(2534),(1253),(12534),(14253),(14)(253)\},\\
S_{33}&=\{(35),(24)(35),(12)(35),(124)(35),(142)(35),(14)(35)\},\\
S_{34}&=\{(345),(2453),(12)(345),(12453),(14532),(14531)\},\\
S_{41}&=\{(15432),(15421),(1543),(154),(15423),(23)(154)\},\\
S_{42}&=\{(2543),(254),(12543),(1254),(13)(254),(13254)\},\\
S_{43}&=\{(354),(2354),(12)(354),(12354),(13542),(1354)\},\\
S_{44}&=\{(45),(23)(45),(12)(45),(123)(45),(132)(45),(13)(45)\}.
\end{split}
\ee
We see that all classes possess six elements.

\subsection{Operators $\opE_{ij}^{\alpha}$}
\label{A2}
\begin{notation}
Let $G$ be a finite group of order $\left\vert G\right\vert =n$ which has $r$
classes of conjugated elements. Then $G$ has exactly $r$ inequivqlent,
irreducible representations, in particular $G$ has exactly $r$ inequivqlent,
irreducible matrix representations. Let 
\be
D^{\alpha }:G\rightarrow \sH(V^{\alpha }),\qquad \alpha =1,2,....,r,\qquad
\dim V^{\alpha }=w^{\alpha }
\ee
be all inequivalent, irreducible representations of $G$ and let chose these
representations to be all unitary (always possible) i.e. 
\be
D^{\alpha }(g)=(D_{ij}^{\alpha }(g)),\qquad 
(D_{ij}^{\alpha }(g))^{+}=(D_{ij}^{\alpha }(g))^{-1},
\ee
where $i,j=1,2,....,w^{\alpha }$ and $V^{\alpha }$ are corresponding representation spaces.
\end{notation}

Similarly by  $D:G\rightarrow \sH(V)$ we denote the complex
finite-dimensional representation of the finite group $G$ in a complex
linear space $V$ ($\dim V=w)$, 

\begin{definition}
Let $D:G\rightarrow Hom(V)$ be an unitary representation of a finite group $G
$ and let $D^{\alpha }:G\rightarrow \sH(V^{\alpha })$ be all inequivalent,
irreducible representations of $G$ . Define 
\be
\opE_{ij}^{\alpha }=\frac{d_{\alpha }}{n!}\sum_{g\in G}D_{ij}^{\alpha
}(g)D(g), 
\ee
where $\alpha =1,2,...,r,\quad i,j=1,2,..,w^{\alpha },\quad
\opE_{ij}^{\alpha }\in \sH(V)$.
\end{definition}

\bigskip 

For these operators we have 

\begin{theorem}
The operators $\opE_{ij}^{\alpha }$ satisfie the following matrix composition
rule%
\be
\opE_{ij}^{\alpha }\opE_{kl}^{\beta }=\delta ^{\alpha \beta }\delta
_{jk}\opE_{il}^{\alpha },
\ee
\ \ \ \ in particular $\opE_{ii}^{\alpha }$ are orthogonal projections.
\end{theorem}

\begin{proof}
\be
\opE_{ij}^{\alpha }\opE_{kl}^{\beta }=\frac{(d_{\alpha })^{2}}{(n!)^{2}}(\sum_{g\in
G}D_{ij}^{\alpha }(g)D(g))(\sum_{h\in G}D_{kl}^{\beta }(h)D(h))
=\frac{(d_{\alpha })^{2}}{(n!)^{2}}\sum_{g,h\in G}D_{ij}^{\alpha
}(g)D_{kl}^{\beta }(h)D(gh).
\ee
Now if we substitute $gh=u\Rightarrow h=g^{-1}u$ we get
\be
\begin{split}
\opE_{ij}^{\alpha }\opE_{kl}^{\beta }&=\frac{(d_{\alpha })^{2}}{(n!)^{2}}\sum_{u,g\in G}\sum_{p}D_{ij}^{\alpha
}(g)D_{kp}^{\alpha }(g^{-1})D_{pl}^{\beta }(u)D(u)=\\
&=\frac{(d_{\alpha })^{2}}{(n!)^{2}}\sum_{p}(\sum_{g\in G}D_{ij}^{\alpha
}(g)D_{kp}^{\beta }(g^{-1}))\sum_{u\in G}D_{pl}^{\alpha }(u)D(u).
\end{split}
\ee
Now we use the orthogonality relation for irreps $D^{\alpha
}(g)=(D_{ij}^{\alpha }(g))$%
\be
\frac{1}{n!}\sum_{h\in G}D_{ij}^{\alpha }(g)D_{kp}^{\beta }(g^{-1})=\frac{1}{%
d_{\alpha }}\delta ^{\alpha \beta }\delta _{ip}\delta _{jk}.
\ee
which leads to%
\be
\begin{split}
\opE_{ij}^{\alpha }\opE_{kl}^{\beta }&=\frac{d_{\alpha }}{n!}\sum_{p}(\delta
^{\alpha \beta }\delta _{ip}\delta _{jk})\sum_{u\in G}D_{pl}^{\alpha
}(u)D(u)=\delta ^{\alpha \beta }\delta _{jk}\sum_{u\in G}D_{il}^{\alpha
}(u)D(u)=\\
&=\delta ^{\alpha \beta }\delta _{jk}\sum_{u\in G}D_{il}^{\alpha
}(u)D(u)=\delta ^{\alpha \beta }\delta _{jk}\opE_{il}^{\alpha }.
\end{split}
\ee
\end{proof}

\begin{remark}
From this proof it follows that the matrix multiplication rule that satisfies
operators $\opE_{il}^{\alpha }$ is the consequence of the fundamental
orthogonality relation for irreps $D^{\alpha }(g)=(D_{ij}^{\alpha }(g)).$
\end{remark}

\end{document}